\documentclass[12pt]{article}

\usepackage{amsmath,amsfonts,amssymb,amsthm}

\numberwithin{equation}{section}

\begin{document}

\newtheorem{theorem}{Theorem}[section]
\newtheorem{corollary}[theorem]{Corollary}
\newtheorem{lemma}[theorem]{Lemma}
\newtheorem{sublemma}[theorem]{Sublemma}
\newtheorem{proposition}[theorem]{Proposition}
\newtheorem{conjecture}[theorem]{Conjecture}

\newcommand{\be}{\begin{equation}}
\newcommand{\ee}{\end{equation}}
\newcommand{\bea}{\begin{eqnarray}}
\newcommand{\eea}{\end{eqnarray}}
\newcommand{\De}{\Delta}
\newcommand{\de}{\delta}
\newcommand{\Z}{{\mathbb Z}}
\newcommand{\N}{{\mathbb N}}
\newcommand{\C}{{\mathbb C}}
\newcommand{\Cs}{{\mathbb C}^{*}}
\newcommand{\R}{{\mathbb R}}
\newcommand{\Q}{{\mathbb Q}}
\newcommand{\T}{{\mathbb T}}
\newcommand{\cW}{{\cal W}}
\newcommand{\cJ}{{\cal J}}
\newcommand{\cE}{{\cal E}}
\newcommand{\cA}{{\cal A}}
\newcommand{\cR}{{\cal R}}
\newcommand{\cP}{{\cal P}}
\newcommand{\cM}{{\cal M}}
\newcommand{\cN}{{\cal N}}
\newcommand{\cI}{{\cal I}}
\newcommand{\cB}{{\cal B}}
\newcommand{\cD}{{\cal D}}
\newcommand{\cC}{{\cal C}}
\newcommand{\cL}{{\cal L}}
\newcommand{\cF}{{\cal F}}
\newcommand{\cH}{{\cal H}}
\newcommand{\cS}{{\cal S}}
\newcommand{\cT}{{\cal T}}
\newcommand{\cU}{{\cal U}}
\newcommand{\cQ}{{\cal Q}}
\newcommand{\cV}{{\cal V}}
\newcommand{\cK}{{\cal K}}
\newcommand{\cO}{{\cal O}}
\newcommand{\intR}{\int_{-\infty}^{\infty}}
\newcommand{\diag}{{\rm diag}}
\newcommand{\Ln}{{\rm Ln}}
\newcommand{\Arg}{{\rm Arg}}
\newcommand{\pz}{\partial_z}
\newcommand{\re}{{\rm Re}\, }
\newcommand{\im}{{\rm Im}\, }
\newcommand{\res}{{\rm Res}}
\newcommand{\intp}{\int_{0}^{\pi/2r}}
\newcommand{\intnp}{\int_{-\pi/2r}^{\pi/2r}}

\title{On Razamat's $A_2$ and $A_3$ kernel identities}
\author{ Simon Ruijsenaars \\ School of Mathematics, \\ University of Leeds, Leeds LS2 9JT, UK}


\maketitle

\begin{abstract}
In recent work on superconformal quantum field theories, Razamat arrived at elliptic kernel identities of a new and striking character: They relate solely to the root systems $A_2$ and $A_3$ and have no coupling type parameters~\cite{Ra18}. The pertinent 2- and 3-variable Hamiltonians are analytic difference operators and the kernel functions are built from the elliptic gamma function. Razamat presented compelling evidence for the validity of these identities, and checked them to a certain order in a power series expansion.  This paper is mainly concerned with analytical proofs of these identities. More specifically, we furnish a complete proof of the identities for the $A_2$ elliptic case and for the $A_3$ hyperbolic case, and consider several specializations. We also discuss the implications the kernel identities might have for a Hilbert space scenario involving common eigenvectors of the Hamiltonians and the integral operators associated with the kernel functions.
\end{abstract}

\tableofcontents


\section{Introduction}

B\"acklund transformations are a well-known ingredient at the classical level of the theory of integrable systems. They have been found and studied for a great many well-known systems, both of the infinite-dimensional and of the finite-dimensional variety. Specializing the latter to the integrable $N$-particle systems of Calogero-Moser and Toda type tied to the root system $A_{N-1}$, 
the associated B\"acklund transformations have a less well known quantum counterpart, namely, so-called kernel identities. 

For a comprehensive study of these finite-dimensional B\"acklund transformations and kernel identities the reader is referred to~\cite{HaRu12}. The latter work contains in particular a survey of the literature in this area. The present paper is concerned with novel and highly unusual kernel identities (due to Razamat~\cite{Ra18}). In particular,  to date no classical counterparts (B\"acklund transformations) of these identities are known. 

We recall that a kernel identity refers to a connection between two Hamiltonians $H_1(v)$ and $H_2(w)$ that are differential or difference operators, with $v$ and $w$ varying over some finite-dimensional vector spaces (not necessarily of the same dimension). Specifically, the Hamiltonians are related by a formula 
\be
H_1(v)K(v,w)=H_2(w)K(v,w),
\ee
where $K(v,w)$ is a function referred to as the kernel function.  

 At the highest level of the Calogero-Moser $A_{N-1}$ type hierarchy, namely the relativistic elliptic quantum level (which includes the Toda case by specialization), the~$N$ commuting Hamiltonians are  analytic difference operators $H_l(x)$, $x\in\C^N$, $l=1,\ldots,N$, with coefficients whose building blocks are essentially Jacobi theta functions. After a similarity transformation with (the square root of)  a factorized weight function, they yield A$\De$Os (analytic difference operators) $A_l(x)$ satisfying a kernel identity of the form
\be\label{kid}
A_l(v)\cS(v,w)=A_l(-w)\cS(v,w),\ \ \ v,w\in\C^N,
\ee
\be\label{kf}
\cS(v,w)\equiv \prod_{j,k=1}^N G(\pm (v_j-w_k)-b),\ \ \ b\in\C.
\ee
Here and below, we use the notation
\be
f(\pm z +c)\equiv f(z+c)f(-z+c),\ \ \ z,c\in\C,
\ee
$b$ is a parameter of coupling constant character, and $G(z)$ is the elliptic gamma function introduced in~\cite{Ru97}. This function and its hyperbolic, trigonometric and rational specializations are crucial ingredients for this paper. (The latter are basically the same as the double sine, $q$-gamma and Euler gamma functions, resp.)

We have collected some salient features of these gamma functions in Appendix~A.
As recalled there, the elliptic and hyperbolic ones have a symmetric dependence on scale parameters $a_+$ and $a_-$, nowadays referred to as modular invariance. However, the dependence on these parameters in the difference operators is not symmetric. Correspondingly, the kernel identities come in pairs: When they hold for an asymmetric difference operator, they also are valid for its modular partner (the operator with $a_+$ and $a_-$ swapped). 

From now on we shall indicate this in our notation by using a subscript $+$ or $-$. (Thus in~\eqref{kid}, for example, we can replace $A_l(x)$ by $A_{l,\de}(x)$ with $\de=+$ or $\de=-$.) By convention, in an A$\De$O $A_+(x)$ the parameter $a_+$ encodes an elliptic period of its coefficients, whereas $a_-$ encodes its shift parameter.

In his recent paper~\cite{Ra18}, Razamat was led to elliptic kernel identities for the root systems $A_2$ and $A_3$ that have a superficial resemblance to the ones just mentioned. His paper was partly inspired by his previous joint work with Gaiotto and Rastelli~\cite{GRS13} and with Zafrir~\cite{RaZa18}. These articles  may be viewed as part of an extensive program to obtain explicit information concerning superconformal field theories, which took off in the eighties of the previous century. In this context, a great many intimate connections to Calogero-Moser and Toda type systems have surfaced, both at the classical and at the quantum level. 

In particular, in recent years the kernel functions and commuting A$\De$Os for the $N$-particle quantum relativistic ($A_{N-1}$) elliptic Calogero-Moser system have been recovered in the setting of index theory for a class of superconformal quantum field theories~\cite{GRS13}. Moreover, the latter setting also gives rise to these objects associated to the quantum elliptic $BC_1$ van Diejen model~\cite{NaRa18}.

Razamat's  A$\De$Os featuring in~\cite{Ra18}   once again involve Jacobi theta function coefficients, yielding a 1-parameter family of commuting modular pairs for the $A_2$ case, and only one modular pair for the $A_3$ case. (It may well be that there exist further independent commuting A$\De$Os satisfying the kernel identities.) With our conventions, his kernel functions are of the form
\be\label{kf2}
\cS_2(v,w,z)\equiv \prod_{k,l,m=1}^3G(v_k+w_l+z_m-\de_2),\ \ \ \de_2\equiv i(a_++a_-)/6,
\ee
in the $A_2$ case, and
\be\label{kf3}
\cS_3(d;v,w)\equiv \prod_{k,l=1}^4G(v_k+w_l-\de_3\pm d),\ \ \ \de_3\equiv i(a_++a_-)/4,\ \ \ d\in\C,
\ee
in the $A_3$ case.  

As they stand, the functions in \eqref{kf2}/\eqref{kf3} depend on vectors with 3/4 components. However, the kernel identities do not hold true in case the vectors  do not satisfy the `center-of-mass' ($A_{N-1}$) restriction
\be
\sum_{j=1}^N x_j=0,\ \ \ N=3,4,
\ee
respectively. Thus $\cS_2$ and $\cS_3$ should be viewed as functions of 6 independent variables, for which we shall choose $v_j,w_j,z_j$, $j=1,2$, and $v_j,w_j$, $j=1,2,3$, resp. 

Another conspicuous difference between~\eqref{kf} and~\eqref{kf2} is the additional vector $z$. Due to the manifest symmetry in $v$, $w$ and $z$, this gives rise to $A_2$ kernel identities  
\be\label{kid2}
A_{2,\de}(\mu;v)\cS_2(v,w,z)=A_{2,\de}(\mu;w)\cS_2(v,w,z)=A_{2,\de}(\mu;z)\cS_2(v,w,z),\ \ \ \de=+,-.
\ee
The A$\De$O family is explicitly given by
\begin{multline}\label{A2sum} 
A_{2,\de}(\mu;x)\equiv \frac{R_{\de}(x_2-x_3\pm\mu)}{R_{\de}(x_1-x_2-ia_{\de}/2)R_{\de}(x_1-x_3-ia_{\de}/2)}\exp\Big(\frac13 ia_{-\de}\big(2\partial_{x_1}-\partial_{x_2}-\partial_{x_3}\big)\Big)
\\
+  \frac{R_{\de}(x_3-x_1\pm\mu)}{R_{\de}(x_2-x_3-ia_{\de}/2)R_{\de}(x_2-x_1-ia_{\de}/2)}\exp\Big(\frac13 ia_{-\de}\big(2\partial_{x_2}-\partial_{x_3}-\partial_{x_1}\big)\Big)
\\
+ \frac{R_{\de}(x_1-x_2\pm\mu)}{R_{\de}(x_3-x_1-ia_{\de}/2)R_{\de}(x_3-x_2-ia_{\de}/2)}\exp\Big(\frac13 ia_{-\de}\big(2\partial_{x_3}-\partial_{x_1}-\partial_{x_2}\big)\Big),
\end{multline}
with the building block $R_{\de}(x)$  detailed in Appendix~A. Note that the second and third summands are obtained by cyclic permutations from the first one.

For the $A_3$ case the kernel identities read
 \be\label{kid3}
A_{3,\de}(v)\cS_3(d;v,w)=A_{3,\de}(w)\cS_3(d;v,w),\ \ \ \de=+,-.
\ee 
Here the A$\De$Os are defined by
\be\label{A3sum}
A_{3,\de}(x)\equiv\sum_{m=1}^4 A_{3,\de}^{(m)}(x),
\ee
with
\be\label{A31}
A_{3,\de}^{(1)}(x)\equiv \frac{R_{\de}(x_2-x_3)R_{\de}(x_3-x_4)R_{\de}(x_4-x_2)}{\prod_{j=2}^4R_{\de}(x_1-x_j-ia_{\de}/2)} \exp\Big(\frac14 ia_{-\de}\big(3\partial_{x_1}-\partial_{x_2}-\partial_{x_3}-\partial_{x_4}\big)\Big),
\ee
and the three summands $A_{3,\de}^{(2)},A_{3,\de}^{(3)}$ and $A_{3,\de}^{(4)}$ obtained from $A_{3,\de}^{(1)}$ by cyclic permutations.
We point out that these operators have no dependence on the parameter~$d$, so here we are dealing with a 1-parameter family of kernel functions.

As they stand, the $A_2$ operators~\eqref{A2sum} and the $A_3$ operators~\eqref{A3sum} have a well-defined action on arbitrary meromorphic functions $f(x)$ with $x\in\C^3$ and $x\in\C^4$, resp. We repeat, however, that in both cases the coordinate sums $\sum_jx_j$ are supposed to vanish, a condition that is clearly left invariant by the argument shifts at hand. 

We proceed by outlining the plan of the paper. Section~2 is concerned with the $A_2$ elliptic identities~\eqref{kid2}. We begin by proving that they do not hold true for unrestricted vectors $v,w,z\in\C^3$, cf.~Prop.~2.1. Imposing next the $A_2$ restriction in the explicit form
\be\label{x3}
 x_3\equiv -x_1-x_2,\ \ \ x=v,w,z,
\ee
we view the two functions at issue (namely $\cL_r$ defined by~\eqref{cLr} and the function~$\cR_r$ obtained by taking $v\leftrightarrow w$)  as meromorphic functions of~$v_1$ (thereby breaking the $S_3$-invariance). The equality of these functions is stated in our main result Theorem~2.2. To prove the equality, we embark on analyzing their poles. They are simple for generic values of the other variables. Combined with a quasi-periodicity property we establish, it suffices to show equality of residue sums at all $v_1$-poles in a quasi-period rectangle. 

This is straightforward for the $w$-independent poles~\eqref{vpind} and is spelled out in Lemma~2.4. By contrast, handling the $w$-dependent poles is a much more laborious enterprise. Even though the remaining symmetries enable us to consider just the single pole~\eqref{vpdep}, the proof of the pertinent Lemma~2.5 is quite long. 

To begin with, we reduce the proof of equality of the two single $v_1$-pole residue functions to a further analysis of their $v_2$-residues at the locations~\eqref{v2L} and~\eqref{v2LR}. 
The former 4 poles are only present in the summands of one of the two residue functions; they are easily verified to yield vanishing residue sums. The 4 poles at~\eqref{v2LR}, however,  occur for summands in both residue functions and involve far more effort. We have relegated their residue analysis to Sublemma~2.6. 

The upshot of this residue analysis is that with the $A_2$ sum constraint in force the asserted equality follows. Hence the elliptic identities~\eqref{kid2} 
hold true, as stated in~Corollary~2.3.
 
Section~3 concerns the $A_3$ case. In Subsection~3.1 we focus on the elliptic version~\eqref{kid3}--\eqref{A31} of the kernel identities. We show that they are not valid for unrestricted $v,w\in\C^4$, cf.~Prop.~3.1. With the $A_3$ sum constraint in effect, however,
the method we used in Section~2 becomes extremely laborious, and we are not even confident that it can be pushed through for the simpler hyperbolic specialization. 

We consider the hyperbolic case in Subsection~3.2. In Prop.~3.2 we first show that the hyperbolic kernel identities are false for unconstrained $v,w\in\C^4$, just like the elliptic ones. 
Insisting next on the sum constraint, we have abolished the method of Section~2 (after going a long way), in favor of a shorter and more insightful proof following a quite different strategy. Unfortunately, we were unable to generalize this method to the elliptic case. The  key idea is to avoid the consideration of $w$-dependent poles in the two functions at hand (once again viewed as meromorphic functions of~$v_1$).  The price to pay for this is that an asymptotic analysis as $\re v_1\to \pm \infty$ is needed, for which no elliptic counterpart exists.

Admittedly, the proof of Theorem~3.3 (which states the hyperbolic kernel identities) is still rather long. In its course, we arrived at three hyperbolic identities that seem far removed from the kernel identities, and that may have independent interest. They pertain to the ratio
\be
R_1\equiv \frac 
{ \cosh (x_2-x_3)\cosh (x_3-x_4)\cosh (x_4-x_2) }{\sinh (x_1-x_2)\sinh (x_1-x_3)\sinh (x_1-x_4)},
\ee
and the three ratios obtained by cyclic permutations. Specifically, the  identities read
\be\label{iden0}
\sum_{m=1}^4 R_m=0,
\ee
\be\label{iden1}
 \sum_{m=1}^4 R_m\cosh (u-2x_m)=\sum_{l=1}^4\sinh (u+2x_l),
\ee
\be\label{iden2}
 \sum_{m=1}^4 R_m\cosh (u-2x_m)\sinh (u+2x_m)^2 
=-\sum_{1\le l_1<l_2<l_3\le 4}\sinh (u+2x_{l_1})\sinh (u+2x_{l_2})\sinh (u+2x_{l_3}),
\ee
where $u,x_1,x_2,x_3,x_4\in\C$ and $\sum_{j=1}^4x_j=0$. We prove equivalent identities in Lemma~3.4.

  Section~4 contains some further observations and developments. In particular, we obtain various results concerning formal self-adjointness, commutativity and regime specializations.  We also present scenarios for the Hilbert space reinterpretation of the kernel identities, inspired by our previous work on this subject. We do this in the form of tentative conjectures. We consider the $A_2$ and $A_3$ cases in Subsections~4.1 and 4.2, respectively.  
  
  In Subsection~4.1 we first show in Prop.~4.1 that the A$\De$Os in the family~\eqref{A2sum} all commute, cf.~\eqref{Acom}. We then show that with suitable restrictions they can be viewed as \emph{formally} positive operators on a weighted $L^2$-space, with weight function explicitly given by~\eqref{Ws}, cf.~Prop.~4.2. It is a challenging open problem to push through the Hilbert space aspects by exploiting the elliptic kernel identities proved in~Section~2. Conjecture~4.3 summarizes the salient properties that would clarify the Hilbert space picture. 
  
Turning to the hyperbolic $A_2$ regime,  we first show in Prop.~4.4 that the hyperbolic kernel identities are not valid for unconstrained $v,w,z$, just like the elliptic ones. Then we proceed along the same lines as in the elliptic case, summarizing our expected Hilbert space scenario in Conjecture~4.5. Next, we tie in the A$\De$Os at issue in this paper with A$\De$Os we encountered in our joint work with Halln\"as~\cite{HaRu12}. The latter arise in the dual regime of the relativistic nonperiodic Toda 3-particle system, cf.~\eqref{BATr}.

For the trigonometric case we can arrive at a counterpart of the kernel identities, but there appears to be no way to reinterpret the pertinent A$\De$O $A_t(x)$~\eqref{Atr} as being self-adjoint on a suitable weighted $L^2$-space. Indeed, it would be a prerequisite to obtain \emph{formal} self-adjointness, and we were not able to find a weight function $W_t(x)$ so that $A_t(x)$ is formally self-adjoint with respect to the measure $W_t(x)dx$.

For the rational A$\De$O $A_r(x)$~\eqref{A2rsum}, obtained as a limit of the hyperbolic or trigonometric one, the kernel identities become trivial: They express that the constant function is a zero-eigenvalue eigenfunction of $A_r(x)$. Essentially the same A$\De$O already arose in~\cite{HaRu12}, namely for the dual nonrelativistic nonperiodic Toda 3-particle system, cf.~\eqref{AATn}.

In Subsection~4.2 we study the $A_3$ case along the same lines as the $A_2$ case. In particular, we show that the A$\De$Os $A_{3,\pm}(x)$ can be reinterpreted as formally positive operators on a weighted $L^2$-space, cf.~Prop.~4.6. We recall that in Subsection~3.1 we have not proved the $A_3$ elliptic kernel identities, but we do expect them to be true. It is very likely that they would be a crucial tool in substantiating the Hilbert space scenario for the elliptic regime, which we encode in Conjecture~4.7.  

For the hyperbolic case we have proved the kernel identities in Subsection~3.2. Conjecture~4.8 summarizes the Hilbert space features we expect in this regime. The above identity~\eqref{iden0} entails that the constant function is a zero-eigenvalue eigenfunction of the hyperbolic A$\De$Os $A_{3,\pm}(x)$ (given by~\eqref{A3sum} and~\eqref{A31h}), but this function is not square-integrable with respect to the pertinent hyperbolic measure.

The trigonometric case again yields kernel identities~\eqref{Atid} for an A$\De$O $A_t(x)$~\eqref{At3} that we cannot even `promote' to a \emph{formally} self-adjoint Hilbert space operator. The limiting kernel identity for the rational A$\De$O $A_r(x)$~\eqref{A3rsum} once more amounts to the constant function being a zero-eigenvalue eigenfunction. This operator is basically equal to an 
  A$\De$O in~\cite{HaRu12}, associated with the dual nonrelativistic nonperiodic Toda 4-particle system, cf.~\eqref{AATn3}.

In Appendix~A we have summarized information about the generalized gamma functions we have occasion to use. It also serves to introduce the elliptic building blocks $R_{\pm}(x)$ and $s_{\pm}(x)$, and to fix conventions and notations.
  

\section{The $A_2$ identities}

 By modular invariance we need only consider one of the modular partners $A_{2,\pm}(\mu;x)$. 
Choosing $A_{2,+}(\mu;x)$, we rewrite this A$\De$O in the shorter form 
\be\label{Ax}
A(x)=\frac{R(x_2-x_3\pm\mu)}{R(x_1-x_2- ia_+/2)R(x_1-x_3- ia_+/2)}\exp\Big(\frac13 ia_-\big(2\partial_1-\partial_2-\partial_3\big)\Big)+\mathrm{cyclic},
\ee
where we have abbreviated $R_+$ as $R$.
From this we readily calculate (by using the difference equations~\eqref{Geades} obeyed by the elliptic gamma function) 
\begin{multline}\label{AS2}
\frac{A(v)\cS_2}{\prod_{k,l,m=1}^3G(v_k+w_l+z_m-ia_-/3-i(a_++a_-)/6)}
\\
=\frac{R(v_2-v_3\pm\mu)}{R(v_1-v_2-t)R(v_1-v_3-t)}
 \prod_{l,m=1}^3R(v_1+w_l+z_m-t+c)+\mathrm{cyclic},
\end{multline}
with  new parameters
\be
t\equiv ia_+/2,\ \ \ c\equiv ia_+/3.
\ee 
We now study whether the function on the right-hand side of~\eqref{AS2} is invariant under swapping $v$ and $w$. 

To this end we divide the  right-hand side by the symmetric product
\be
\prod_{k,l,m=1}^3R(v_k+w_l+z_m-t+c),
\ee
yielding the function
\be\label{cL} 
\cL\equiv \frac{R(v_2-v_3\pm\mu)}{R(v_1-v_2-t)R(v_1-v_3-t)}
 \prod_{k\ne 1,l,m}\frac{1}{R(v_k+w_l+z_m-t+c)}+\mathrm{cyclic}.
\ee
The question is now whether $\cL$ equals
\be\label{cR}
\cR\equiv \frac{R(w_2-w_3\pm\mu)}{R(w_1-w_2-t)R(w_1-w_3-t)}
 \prod_{k,l\ne 1,m}\frac{1}{R(v_k+w_l+z_m-t+c)}+\mathrm{cyclic}.
\ee

At this point it can already be seen that this is false without the $A_2$ constraint
\be\label{A2con}
x_1+x_2+x_3=0,\ \ \ x=v,w,z.
\ee
Equivalently,  the kernel identities are false for unrestricted variables.

\begin{proposition}
Letting $v,w,z\in\C^3$, the functions $\cL$ and $\cR$ are not equal.
\end{proposition}
\begin{proof}
Let us view $\cL-\cR$ as a meromorphic function of $v_1$. For generic values of the remaining variables this function is a sum of 6 terms that have simple poles. Consider now the poles arising for $v_1=v_2+2kt$, $k\in \Z$. They only occur  in the first two summands. Upon multiplication  by $R(v_1-v_2-t)$, their sum becomes
\begin{multline}\label{sumne1}
 \prod_{l,m}\frac{1}{R(v_3+w_l+z_m-t+c)}\Big(\frac{R(v_2-v_3\pm\mu)}{R(v_1-v_3-t)}
 \prod_{l,m}\frac{1}{R(v_2+w_l+z_m-t+c)}
 \\
 -\exp[2ir(v_1-v_2)]\frac{R(v_3-v_1\pm\mu)}{R(v_2-v_3-t)}
 \prod_{l,m}\frac{1}{R(v_1+w_l+z_m-t+c)}\Big),
\end{multline}
where we used~\eqref{Rade}. Clearly, this sum vanishes identically for $v_1=v_2$. But for $v_1=v_2+2t$, the sum in brackets equals 
\begin{multline}\label{sumne2}
 \frac{R(v_2-v_3\pm\mu)}{R(v_2-v_3-t)}
 \prod_{l,m}\frac{1}{R(v_2+w_l+z_m-t+c)}\Big(-\exp[2ir(v_2-v_3)]
 \\
 +\exp\Big[2ir\Big(2t+2(v_3-v_2-t)+9(v_2+c)+3\sum_{l}w_l+3\sum_mz_m\Big)\Big]\Big),
\end{multline}
a function that obviously is nonzero. This nonvanishing residue entails that $\cL$ and $\cR$ are not equal for arbitrary $v,w,z\in\C^3$, as asserted.
\end{proof}

More generally, when we take $v_1\to v_1+2t$, then the 3 summands of $\cL$ pick up multipliers 
\be\label{muj3}
\mu_1 \exp 2ir (2v_1),\ \ \mu_2 \exp 2ir (8v_1),\ \ \mu_3 \exp 2ir (8v_1), 
\ee
whereas in $\cR$ this yields multipliers 
\be\label{nuj} 
\nu_j \exp 2ir (6v_1),\ \ \ j=1,2,3,
\ee
with $\mu_j$ and $\nu_j$ not depending on $v_1$. (As in the proof just given, this can be verified by using~\eqref{Rade}.)

From now on we impose the $A_2$ constraint~\eqref{A2con}, in the more specific form of~\eqref{x3}. 
Then the two transpositions $x_1\leftrightarrow x_2$  and $x_2\leftrightarrow x_3$ generating $S_3$ amount to the maps
\be
(x_1,x_2)\mapsto(x_2,x_1),\ \ \ (x_1,x_2)\mapsto (x_1,-x_1-x_2),
\ee
and the A$\De$O~\eqref{Ax} turns into
\begin{multline}\label{A2r} 
A_r(x)= \frac{R(x_1+2x_2\pm\mu)}{R(x_1-x_2-t)R(2x_1+x_2-t)}\exp\Big(\frac13 ia_-\big(2\partial_{x_1}-\partial_{x_2} \big)\Big)
\\
+\frac{R(2x_1+x_2\pm\mu)}{R(x_1+2x_2-t)R(x_1-x_2+t)}\exp\Big(\frac13 ia_-\big(-\partial_{x_1}+2\partial_{x_2} \big)\Big)
\\
+\frac{R(x_1-x_2\pm\mu)}{R(2x_1+x_2+t)R(x_1+2x_2+t)}\exp\Big(\frac13 ia_-\big(-\partial_{x_1}-\partial_{x_2} \big)\Big) .
\end{multline}

In sharp contrast to the above unequal multipliers for the unconstrained variables,  for all of the 6 summands at issue the shift $v_1\to v_1+2t$ now gives rise to a multiplier
\be\label{eqmult}
\exp[2ir(6ia_++6v_2)]\exp[2ir(12v_1)].
\ee
This state of affairs is not immediate from~\eqref{cL} and~\eqref{cR}, but our variable restriction~\eqref{x3}  turns the 9-variable function $\cL$ into the 6-variable function
\begin{multline}\label{cLr} 
\cL_r= \frac{(v_1+2v_2\pm\mu)}{(v_1-v_2-t)(2v_1+v_2-t)}  \prod_{l,m}\frac{1}{(v_2+w_l+z_m-t+c)(-v_1-v_2+w_l+z_m-t+c)}\\
+\frac{(2v_1+v_2\pm\mu)}{(v_1+2v_2-t)(v_1-v_2+t)} 
 \prod_{l,m}\frac{1}{(-v_1-v_2+w_l+z_m-t+c)(v_1+w_l+z_m-t+c)}
\\
+\frac{(v_1-v_2\pm\mu)}{(2v_1+v_2+t)(v_1+2v_2+t)}
 \prod_{l,m}\frac{1}{(v_1+w_l+z_m-t+c)(v_2+w_l+z_m-t+c)} .
\end{multline}
(Here and from now on in this section, we abbreviate $R_+(x)$, $x\in\C$, to $(x)$.) From this restricted form~$\cL_r$ of $\cL$ and the corresponding form~$\cR_r$ of $\cR$ (obtained by swapping $v$ and $w$), the multiplier assertion~\eqref{eqmult} can be readily verified by using~\eqref{Rade}. 

The equal multiplier property is a key ingredient for showing $\cL_r=\cR_r$. Indeed, it implies that we need only show that the residue sums of $\cL_r-\cR_r$, viewed as a meromorphic and $\pi/r$-periodic function of~$v_1$, are zero for the (generically simple) poles
\be\label{vpind}
v_1=v_2,-2v_2,-v_2/2+\omega_j,\ \ \ j=0,1,2,3,
\ee
with
\be
\omega_0:=0,\ \   \omega_1:=\pi/2r,\ \   \omega_2:=t,\ \   \omega_3:=  t+\pi/2r,
\ee
and for
\be\label{vpdep}
v_1=-w_1-z_1-c,\ \ c=ia_+/3.
\ee
To see this, note that when this holds true, we may invoke $S_3$ symmetry to deduce vanishing residue sums for all of the 9 poles
\be
v_1=-w_l-z_m-c,\ \ l,m=1,2,3,
\ee
and then it follows that $\cL_r-\cR_r$ is an entire and $\pi/r$-periodic function  of~$v_1$ with multiplier~\eqref{eqmult} under the shift $v_1\to v_1+2t$.

It is widely known that this state of affairs implies that $\cL_r-\cR_r$ must vanish identically, but it may be helpful   to add a short proof. 
Introducing new variables
\be
z:=\exp (2irv_1),\ \ \ q:= \exp (-2ra_+), 
\ee
we may and shall view $\cL_r-\cR_r$ as a function $F(z)$ that is analytic in $\C^{*}$ and that satisfies
\be\label{Fq}
F(qz)=Cz^{12} F(z),\ \ \ C:=\exp[2ir(6ia_++6v_2)].
\ee
As a consequence, $F(z)$ admits a representation as a Laurent series
\be
F(z)=\sum_{k=-\infty}^{\infty} a_kz^k,
\ee
which converges in $\C^*$. This implies in particular
\be\label{az}
\lim_{k\to\pm\infty}a_k=0.
\ee
Now from~\eqref{Fq} we readily obtain
\be
a_k=Cq^{-k} a_{k-12}.
\ee
Since we have $|q|<1$, the assumption that there exists a nonzero $a_{k_0}$ leads to a contradiction to~\eqref{az}: It entails that the sequence $a_{k_0+12m}$ diverges for $m\to\infty$.

We are now prepared to state the main results of this section.

\begin{theorem}
The function $\cL_r$ defined by~\eqref{cLr} and the function~$\cR_r$ obtained by swapping $v$ and $w$ are equal.
\end{theorem}
\begin{corollary}
With the constraint
\be
\sum_{j=1}^3 x_j=0,\ \ \ x=v,w,z,
\ee
in effect, the $A_2$ elliptic kernel identities~\eqref{kid2} hold true.
\end{corollary}

This corollary is a direct consequence of the theorem and modular invariance.  We prove the theorem by detailing the vanishing of the residue sums at the   poles~\eqref{vpind} that do not depend on $w$ and~$z$ in Lemma~2.4, and at the pole~\eqref{vpdep} in Lemma~2.5. As already shown, this suffices to deduce Theorem~2.2. 

\begin{lemma}
The residue sums of $\cL_r-\cR_r$ at the 6 poles~\eqref{vpind} are zero.
\end{lemma}
\begin{proof} Recalling our convention $(x):=R_+(x)$, we successively consider the 6 cases at issue.

($\mathbf {v_1=v_2}$) We multiply by $(v_1-v_2-t)$ and then take $v_1\to v_2$ to get
\begin{multline}
\frac{(3v_2\pm\mu)}{(3v_2-t)}
 \prod_{l,m}\frac{1}{(v_2+w_l+z_m-t+c)(-2v_2+w_l+z_m-t+c)}
\\
-  \frac{(3v_2\pm\mu)}{(3v_2-t)}
 \prod_{l,m}\frac{1}{(-2v_2+w_l+z_m-t+c)(v_2+w_l+z_m-t+c)},
\end{multline}
which vanishes.

($\mathbf{v_1=-2v_2}$) We multiply by $(v_1+2v_2-t)$ and then take $v_1\to -2v_2$ to get
\begin{multline}
\frac{(-3v_2\pm\mu)}{(-3v_2+t)}
 \prod_{l,m}\frac{1}{(v_2+w_l+z_m-t+c)(-2v_2+w_l+z_m-t+c)}
\\
-  \frac{(-3v_2\pm\mu)}{(-3v_2+t)}
 \prod_{l,m}\frac{1}{(-2v_2+w_l+z_m-t+c)(v_2+w_l+z_m-t+c)}=0.
\end{multline}

($\mathbf{v_1=-v_2/2+\omega_j}$) We multiply by $(2v_1+v_2-t)$ and then take $v_1\to -v_2/2+\omega_j$, yielding
\begin{multline}\label{4pcal}
\frac{(3v_2/2+\omega_j\pm\mu)}{(-3v_2/2+\omega_j-t)}
 \prod_{l,m}\frac{1}{(v_2+w_l+z_m-t+c)(-v_2/2-\omega_j+w_l+z_m-t+c)}
\\
- \exp[2ir(2\omega_j )]\frac{(-3v_2/2+\omega_j\pm\mu)}{(3v_2/2+\omega_j+t)}
\\
\times  \prod_{l,m}\frac{1}{(-v_2/2+\omega_j+w_l+z_m-t+c)(v_2+w_l+z_m-t+c)}.
\end{multline}
This expression clearly vanishes for $j=0$ and also for $j=1$ (using $\pi/r$-periodicity). For $j=2$ it is proportional to
\be\label{p2cal}
1-\exp\Big[2ir\Big((2t)+(3v_2)+(3v_2/2+t)+9(-v_2/2-t+c)+3\sum_lw_l+3\sum_mz_m\Big)\Big].
\ee
Using $-6t+9c=0$ and the sum constraint~\eqref{A2con}, we infer this multiplier vanishes. Likewise, using once more $\pi/r$-periodicity, we get a vanishing residue sum for $j=3$. 
\end{proof}

\begin{lemma}
The residue sum of $\cL_r-\cR_r$ at the  pole~\eqref{vpdep} vanishes.
\end{lemma}

\begin{proof} We multiply $\cL_r$~\eqref{cLr} and $\cR_r$ by $(v_1+w_1+z_1-t+c)$ and then let $v_1$ converge to $-w_1-z_1-c$. Pulling out the obvious common factors, this yields for~$\cL_r$
\begin{multline}\label{Lsub}
\prod_{(l,m)\ne (1,1)} (-w_1-z_1+w_l+z_m-t)^{-1}
\\
\times\Big(\frac{(-2w_1-2z_1-2c+v_2\pm\mu)}{(-w_1-z_1-c+2v_2-t)(-w_1-z_1-c-v_2+t)}
\\
\times \prod_{l,m} (-w_1-z_1+v_2-w_l-z_m+t-2c)^{-1}
\\
+\frac{(-w_1-z_1-c-v_2\pm\mu)}{(-2w_1-2z_1-2c+v_2+t)(-w_1-z_1-c+2v_2+t)} 
\\
\times\prod_{l,m} (v_2+w_l+z_m-t+c)^{-1}\Big),
\end{multline}
and for $\cR_r$
\begin{multline}\label{Rsub}
\prod_{m>1}(-z_1+z_m-t)^{-1}\prod_m \Big[ (w_1+v_2+z_m-t+c)(2w_1+z_1-v_2+z_m-t+2c) \Big]^{-1}
\\
\times\Big(\frac{(2w_1+w_2\pm\mu)}{(w_1+2w_2-t)(w_1-w_2+t)}
\\
\times\prod_m [(2w_1+w_2+z_1-z_m+t)(w_1+w_2-v_2-z_m+t-c)(w_2-z_1+v_2-z_m+t-2c)]^{-1}
\\
+\frac{(w_1-w_2\pm\mu)}{(2w_1+w_2+t)(w_1+2w_2+t)}
\\
\times\prod_m [(w_2-w_1-z_1+z_m-t)(w_2+v_2+z_m-t+c)(w_2+w_1+z_1-v_2+z_m-t+2c)]^{-1}\Big).
\end{multline}

In order to pull out further common factors, we now recall $t=ia_+/2,c=ia_+/3$, 
to obtain
\be\label{x2c}
(x+t-2c)= (x-ia_+/6)=(x-t+c).
\ee
Using this we can rewrite~\eqref{Lsub} as
\begin{multline}\label{Lsub2}
\prod_{m>1}(-z_1+z_m-t)^{-1}\prod_{m} (-w_1-z_1+w_2+z_m-t)^{-1}(-w_1-z_1+w_3+z_m-t)^{-1}
\\
\times\Big(\frac{(-2w_1-2z_1-2c+v_2\pm\mu)}{(-w_1-z_1+2v_2-t-c)(-w_1-z_1-v_2+t-c)} \prod_{l,m} (v_2-w_1-z_1-w_l-z_m-t+c)^{-1}
\\
+\frac{(-w_1-z_1-c-v_2\pm\mu)}{(-2w_1-2z_1+v_2-t+c)(-w_1-z_1+2v_2+t-c)} \prod_{l,m} (v_2+w_l+z_m-t+c)^{-1}\Big).
\end{multline}
Consider next the factors with $l=m$. A moment's thought shows they give rise to 4 common factors, yielding
\begin{multline}\label{Lsub3}
\Big(\prod_{m>1}(-z_1+z_m-t)\prod_{m} (-w_1-z_1+w_2+z_m-t)(-w_1-z_1+w_3+z_m-t)\Big)^{-1}
\\
\times\Big( (v_2-2w_1-2z_1-t+c)\prod_k(v_2+w_k+z_k-t+c) \Big)^{-1}
\\
\times\Big(\frac{(-2w_1-2z_1-2c+v_2\pm\mu)}{(-w_1-z_1+2v_2-t-c) }  \prod_{l\ne m} (v_2-w_1-z_1-w_l-z_m-t+c)^{-1}
\\
+\frac{(-w_1-z_1-c-v_2\pm\mu)}{ (-w_1-z_1+2v_2+t-c)} \prod_{l\ne m} (v_2+w_l+z_m-t+c)^{-1}\Big).
\end{multline}
Finally, we can match the two factors with $(l,m)=(2,3),(3,2)$, which gives
\begin{multline}\label{Lsub4}
\Big(\prod_{m>1}(-z_1+z_m-t)\prod_{m} (-w_1+w_2-z_1+z_m-t)(-w_1+w_3-z_1+z_m-t)\Big)^{-1}
\\
\times\Big((v_2+w_2+z_3-t+c)(v_2+w_3+z_2-t+c) (v_2-2w_1-2z_1-t+c)\prod_k(v_2+w_k+z_k-t+c) \Big)^{-1}\times S_{\cL},
\end{multline}
with
\begin{multline}\label{SL} 
S_{\cL}\equiv \frac{(v_2-2w_1-2z_1-2c\pm\mu)}{(2v_2-w_1-z_1-t-c) }
\\
\times\Big( (v_2-2w_1+z_3-t+c) (v_2+w_3-2z_1-t+c) (v_2-2w_1+z_2-t+c) (v_2+w_2-2z_1-t+c)\Big)^{-1}
\\
+\frac{(v_2+w_1+z_1+c\pm\mu)}{ (2v_2-w_1-z_1+t-c)}
\\
\times \Big(  (v_2+w_1+z_2-t+c)(v_2+w_2+z_1-t+c)(v_2+w_1+z_3-t+c)(v_2+w_3+z_1-t+c)\Big)^{-1}.
\end{multline}

Next, we proceed in a similar way for \eqref{Rsub}. Thus we first get, using~\eqref{x2c},
\begin{multline}\label{Rsub2}
\Big(\prod_{m>1}(-z_1+z_m-t)\prod_m  (v_2+w_1+z_m-t+c)(v_2-2w_1-z_1-z_m-t+c) \Big)^{-1}
\\
\times\Big(\frac{(w_1-w_3\pm\mu)}{(w_2-w_3-t)(w_1-w_2+t)}
\\
\times\prod_m [(w_1-w_3+z_1-z_m+t)(v_2+w_3+z_m-t+c)(v_2+w_2-z_1-z_m-t+c)]^{-1}
\\
+\frac{(w_1-w_2\pm\mu)}{(w_1-w_3+t)(w_2-w_3+t)}
\\
\times\prod_m [(w_2-w_1-z_1+z_m-t)(v_2+w_2+z_m-t+c)(v_2+w_3-z_1-z_m-t+c)]^{-1}\Big).
\end{multline}
Consider the 3 factors in the 2 products over~$m$. We can match the $m=1$ case of the first factor and the $m>1$ cases of the second and third factors. Then we get
\begin{multline}\label{Rsub3}
\Big(\prod_{m>1}(-z_1+z_m-t)\prod_m  (v_2+w_1+z_m-t+c)(v_2-2w_1-z_1-z_m-t+c) \Big)^{-1}
\\
\times\Big( (w_1-w_2+t) (w_1-w_3+t)\prod_{l,m>1} (v_2+w_l+z_m-t+c)  \Big)^{-1}\times S_{\cR},
\end{multline}
with
\begin{multline}\label{SR}
S_{\cR}\equiv \frac{(w_1-w_3\pm\mu)}{(w_2-w_3-t)}
\\
\times\Big[\prod_{m>1} (w_1-w_3+z_1-z_m+t)\cdot (v_2+w_3+z_1-t+c)(v_2+w_2-2z_1-t+c)\Big]^{-1}
\\
+\frac{(w_1-w_2\pm\mu)}{(w_2-w_3+t)}
\\
\times\Big[\prod_{m>1} (w_1-w_2+z_1-z_m+t)\cdot (v_2+w_2+z_1-t+c)(v_2+w_3-2z_1-t+c)\Big]^{-1}.
\end{multline}

Next, we compare the factors in the multipliers~$\mu_{\cL}$ of $S_{\cL}$ and~$\mu_{\cR}$ of $S_{\cR}$ in \eqref{Lsub4} and \eqref{Rsub3}, resp. All of the 6 $v_2$-dependent factors in $\mu_{\cL}$ are matched by factors in~$\mu_{\cR}$. We can also match 4 of the remaining factors. Therefore, we are left with the question whether 
\be\label{EL}
E_{\cL}\equiv \Big(\prod_{m>1} (-w_1+w_2-z_1+z_m-t)(-w_1+w_3-z_1+z_m-t)\Big)^{-1}S_{\cL} 
\ee
and
\be\label{ER}
E_{\cR}\equiv \Big(\prod_{m>1}   (v_2+w_1+z_m-t+c)(v_2-2w_1+z_m-t+c) \Big)^{-1}S_{\cR}
\ee
are equal. 

To answer this, we view $E_{\cL}-E_{\cR}$ as a meromorphic $\pi/r$-periodic function of $v_2$ and calculate the 4 multipliers arising under the shift $v_2\to v_2+2t$. The result is that they are all equal to
\be
\exp[2ir(-3w_1-3z_1+2ia_+)]\exp[2ir(6v_2)].
\ee
Thus it follows as before that we need only inspect the residues at the poles 
\be\label{v2L}
v_2=(w_1+z_1+c)/2+\omega_j,\ \ \ j=0,1,2,3, 
\ee
(which occur solely in the summands of $S_{\cL}$), and at
the poles
\be\label{v2LR}
v_2+c=-w_2-z_1,-w_2+2z_1, 2w_1-z_2,-w_1-z_2.
\ee
(Indeed, both $E_{\cL}$ and $E_{\cR}$ are invariant under swapping $w_2$, $w_3$, and $z_2$, $z_3$.)

Using $-t+3c/2=0$, it is straightforward to check residue cancellation at the 4 poles~\eqref{v2L} (cf.~the calculations~\eqref{4pcal}--\eqref{p2cal} in Lemma~2.4).  We can therefore complete the proof of the lemma by showing equality of
the residues of $E_{\cL}$ and $E_{\cR}$ at the remaining 4 poles~\eqref{v2LR}. Unfortunately, this involves longer arguments, which we relegate to the following sublemma. 

\end{proof}

\begin{sublemma} 
The residues of $E_{\cL}$ and $E_{\cR}$ at the poles~\eqref{v2LR} are equal.
\end{sublemma}
\begin{proof}
 We successively consider the residues for the 4 poles at hand. 
 \vspace{6mm}

$\mathbf{(v_2=-w_2-z_1-c)}$ For $E_{\cL}$ we obtain from~\eqref{EL} and~\eqref{SL}
(upon multiplication by $(v_2+w_2+z_1-t+c)$ and then setting $v_2=-w_2-z_1-c$):
\begin{multline}
 \Big(\prod_{m>1} (-w_1+w_2-z_1+z_m-t)(-w_1+w_3-z_1+z_m-t)\Big)^{-1}
 \\
\times \frac{(-w_2+w_1\pm\mu)}{ (-2w_2-w_1-3z_1+t-3c)}
\\
\times \Big(  (-w_2+w_1-z_1+z_2-t) (-w_2+w_1-z_1+z_3-t)(-w_2+w_3-t)\Big)^{-1},
\end{multline}
whereas from~\eqref{ER} and~\eqref{SR} we get
\begin{multline}
\Big(\prod_{m>1}   (-w_2+w_1-z_1+z_m-t)(-w_2-2w_1-z_1+z_m-t) \Big)^{-1}
\\
\times\frac{(w_1-w_2\pm\mu)}{(w_2-w_3+t)}
\\
\times\Big[\prod_{m>1} (w_1-w_2+z_1-z_m+t)\cdot  (-w_2+w_3-3z_1-t)\Big]^{-1}.
\end{multline}
  Recalling $t-3c=-t$, we see that these two expressions are equal.
\vspace{6mm}

$\mathbf{(v_2=-w_2+2z_1-c)}$ For $E_{\cL}$ we obtain
\begin{multline}\label{r2L}
 \Big(\prod_{m>1} (-w_1+w_2-z_1+z_m-t)(-w_1+w_3-z_1+z_m-t)\Big)^{-1}
 \\
\times \frac{(-w_2-2w_1-3c\pm\mu)}{ (-2w_2-w_1+3z_1-t-3c)}
\\
\times \Big(  (-w_2-2w_1+2z_1+z_3-t) (-w_2+w_3-t)(-w_2-2w_1+2z_1+z_2-t)\Big)^{-1},
\end{multline}
and for $E_{\cR}$ we get
\begin{multline}\label{r2R}
\Big(\prod_{m>1}   (-w_2+w_1+2z_1+z_m-t)(-w_2-2w_1+2z_1+z_m-t) \Big)^{-1}
\\
\times\frac{(w_1-w_3\pm\mu)}{(w_2-w_3-t)}
\\
\times\Big[\prod_{m>1} (w_1-w_3+z_1-z_m+t)\cdot  (-w_2+w_3+3z_1-t)\Big]^{-1}.
\end{multline}

Now we can rewrite the first product on the first line of~\eqref{r2R} as
\be\label{1R}
\frac{1}{ \prod_{m>1}   (w_1-w_2+z_1-z_m-t) }
=\frac{\exp\Big(-2ir\sum_{m>1}(w_1-w_2+z_1-z_m)\Big)}{\prod_{m>1}   (w_1-w_2+z_1-z_m+t) },
\ee
which equals the first product on the first line of~\eqref{r2L} save for the exponential. Next, we rewrite the second and third line of~\eqref{r2L} as
\begin{multline}\label{23L}
 \frac{(w_1-w_3+ia_+\mp\mu)}{ (-w_2+w_3+3z_1-t-ia_+)}
\\
\times \Big(  (-w_2-2w_1+z_1-z_2-t) (w_2-w_3+t)(-w_2-2w_1+z_1-z_3-t)\Big)^{-1}
\\
=-\exp[-2ir((2w_1-2w_3+ia_+)+(-w_2+w_3+3z_1-t-ia_+/2))] \frac{(w_1-w_3\mp\mu)}{ (w_2-w_3-3z_1+t)}
\\
\times(-)\exp[ 2ir(w_2-w_3)]
 \Big( (w_2-w_3-t)\prod_{m>1} (-w_2-2w_1+z_1-z_m-t)  \Big)^{-1}
\\
=\exp[-2ir(2w_1-2w_2+3z_1)]\frac{(w_1-w_3\mp\mu)}{ (w_2-w_3-3z_1+t)}
\\
\times \Big( (w_2-w_3-t)\prod_{m>1} (-w_2-2w_1+2z_1+z_m-t)  \Big)^{-1}.
\end{multline}
The product in the last line occurring here matches the second product on the first line of~\eqref{r2R}. The remaining $R_+$-factors match as well. Finally, the exponential multipliers in~\eqref{23L} and~\eqref{1R} are equal. Therefore, the residues are equal.
\vspace{6mm}

$\mathbf{(v_2=2w_1-z_2-c)}$ For $E_{\cL}$ we obtain
\begin{multline}\label{r3L}
 \Big(\prod_{m>1} (-w_1+w_2-z_1+z_m-t)(-w_1+w_3-z_1+z_m-t)\Big)^{-1} 
 \\
\times \frac{(-2z_1-z_2-3c\pm\mu)}{(3w_1-z_1-2z_2-t-3c)}\Big( (-z_2+z_3-t)(2w_1+w_3-2z_1-z_2-t)(2w_1+w_2-2z_1-z_2-t)\Big)^{-1},
 \end{multline}
 and for $E_{\cR}$ we get
 \begin{multline}\label{r3R}
 \Big((-z_2+z_3-t)\prod_{m>1}(3w_1-z_2+z_m-t)\Big)^{-1}
 \\
\times\Big[ \frac{( w_1-w_3\pm\mu)}{(w_2-w_3-t)}\Big(\prod_{m>1}(w_1-w_3+z_1-z_m+t)\cdot (2w_1+w_3+z_1-z_2-t)(2w_1+w_2-2z_1-z_2-t)\Big)^{-1}
\\
+\frac{(w_1-w_2\pm\mu)}{(w_2-w_3+t)}\Big(\prod_{m>1}(w_1-w_2+z_1-z_m+t)\cdot (2w_1+w_2+z_1-z_2-t)(2w_1+w_3-2z_1-z_2-t)\Big)^{-1}\Big].
\end{multline}
Using
\be\label{simL}
(3w_1-z_1-2z_2-t-3c)=-\exp[2ir(3w_1-z_1-2z_2-t-ia_+/2)](3w_1-z_1-2z_2-t)
\ee
in~\eqref{r3L}, and
\be\label{simR1}
(2w_1+w_2+z_1-z_2-t)=-\exp[2ir(w_1-w_3+z_1-z_2)](w_1-w_3+z_1-z_2+t),
\ee
\be\label{simR2}
(2w_1+w_3+z_1-z_2-t)=-\exp[2ir(w_1-w_2+z_1-z_2)](w_1-w_2+z_1-z_2+t),
\ee
in~\eqref{r3R}, we obtain common denominator factors
\be
(3w_1-z_1-2z_2-t)
(-z_2+z_3-t)(-w_1+w_2-z_1+z_2-t)(-w_1+w_3-z_1+z_2-t).
\ee
Canceling these factors, we are left with
\begin{multline}
-\exp[-2ir(3w_1-z_1-2z_2-2t)+2ir(-4z_1-2z_2-2t)]
\\
\times(-2z_1-z_2\pm\mu)\Big(\prod_{l>1}(-w_1+w_l-z_1+z_3-t)(2w_1+w_l-2z_1-z_2-t)\Big)^{-1},
\end{multline}
and
\begin{multline}
\frac{-1}{(3w_1-t)}\Big[\exp[-2ir(w_1-w_2+z_1-z_2)]\frac{(w_1-w_3\pm\mu)}{(w_2-w_3-t)}
\\
\times\Big((w_1-w_3+z_1-z_3+t)(2w_1+w_2-2z_1-z_2-t)\Big)^{-1}
\\
+\exp[-2ir(2w_1+w_2+z_1-z_2)]\frac{(w_1-w_2\pm\mu)}{(w_2-w_3+t)}
\\
\times\Big((w_1-w_2+z_1-z_3+t)(2w_1+w_3-2z_1-z_2-t)\Big)^{-1}\Big].
\end{multline}

Next, we multiply these formulas by
\be
-(3w_1-t)(w_2-w_3-t)\prod_{l>1}(-w_1+w_l-z_1+z_3-t)(2w_1+w_l-2z_1-z_2-t),
\ee
so that we get
\be
\exp[2ir(-3w_1-3z_1)](3w_1-t)(w_2-w_3-t)(-2z_1-z_2\pm\mu),
\ee
and
\begin{multline}
\exp[-2ir(w_1-w_2+z_1-z_2)](w_1-w_3\pm\mu)
\\
\times(-w_1+w_2-z_1+z_3-t)(2w_1+w_3-2z_1-z_2-t)
\\
-\exp[-2ir(2w_1+w_2+z_1-z_2)+2ir(w_2-w_3)](w_1-w_2\pm\mu)
\\
\times
(-w_1+w_3-z_1+z_3-t)(2w_1+w_2-2z_1-z_2-t).
\end{multline}
The exponential multipliers in the second expression are equal. Hence we are left with showing equality of
\be
\exp[(2ir(-2w_1-w_2-2z_1-z_2)](3w_1-t)(w_1+2w_2-t)(2z_1+z_2\pm\mu),
\ee
and
\begin{multline}
(2w_1+w_2\pm\mu)(-w_1+w_2-2z_1-z_2-t)(w_1-w_2-2z_1-z_2-t)
\\
-(w_1-w_2\pm\mu)(-2w_1-w_2-2z_1-z_2-t)(2w_1+w_2-2z_1-z_2-t).
\end{multline}

To prove equality, we introduce new variables
\be
x=2w_1+w_2,\ \ \ y=w_1-w_2,\ \ \ z=2z_1+z_2.
\ee
Viewing the result as yielding two functions of~$z$, these functions read
\be
l(z):=\exp[2ir(-x-z)](x+y-t)(x-y-t)(z\pm\mu),
\ee
and
\be
r(z):=(x\pm\mu)(-y-z-t)(y-z-t)-(y\pm\mu)(-x-z-t)(x-z-t).
\ee
We have
\be
l(z-2t)=\exp[2ir(2z-2t)+2ir(2t)]l(z),
\ee
and
\be
r(z-2t)=\exp[2ir(2z)]r(z),
\ee
so the quotient
\be
q(z):=r(z)/l(z)
\ee
is elliptic with periods $\pi/r$, $ia_+$. The residue sums at the poles $z=\mu-t$ and $z=-\mu-t$ clearly vanish, so $q(z)$ is constant in~$z$. Taking $z=-x$, we obtain $q(z)=1$. This implies equality of residues at the third pole.
\vspace{6mm}

$\mathbf{(v_2=-w_1-z_2-c)}$ For $E_{\cL}$ we now obtain
\begin{multline}\label{r4L}
 \Big(\prod_{m>1} (-w_1+w_2-z_1+z_m-t)(-w_1+w_3-z_1+z_m-t)\Big)^{-1} 
 \\
\times \frac{(z_1-z_2\pm\mu)}{(-3w_1-z_1-2z_2+t-3c)}\Big((w_1-w_2-z_1+z_2+t)(-z_2+z_3-t)(w_1-w_3-z_1+z_2+t)\Big)^{-1},
 \end{multline}
 and for $E_{\cR}$ we get
 \begin{multline}\label{r4R}
 \Big((-z_2+z_3-t)\prod_{m>1}(-3w_1-z_2+z_m-t)\big)^{-1}
 \\
\times\Big[ \frac{( w_1-w_3\pm\mu)}{(w_2-w_3-t)}\Big(\prod_{m>1}(w_1-w_3+z_1-z_m+t)\cdot (w_1-w_3-z_1+z_2+t)(w_1-w_2+2z_1+z_2+t)\Big)^{-1}
\\
+\frac{(w_1-w_2\pm\mu)}{(w_2-w_3+t)}\Big(\prod_{m>1}(w_1-w_2+z_1-z_m+t)\cdot (w_1-w_2-z_1+z_2+t)(w_1-w_3+2z_1+z_2+t)\Big)^{-1}\Big].
\end{multline}
Here we have common denominator factors
\be
(-3w_1-z_1-2z_2-t)
(-z_2+z_3-t)(-w_1+w_3-2z_1-z_2-t)(-w_1+w_2-2z_1-z_2-t).
\ee
Canceling these factors, we are left with
\begin{multline}
 \Big(\prod_{l>1} (-w_1+w_l-z_1+z_2-t)\Big)^{-1} 
  (z_1-z_2\pm\mu)
 \\
\times \Big( (-w_1+w_2+z_1-z_2-t)(-w_1+w_3+z_1-z_2-t)\Big)^{-1},
\end{multline}
and
\begin{multline}
\frac{1}{(-3w_1-t)}\Big[ \frac{(w_1-w_3\pm\mu)}{(w_2-w_3-t)}
 \Big((w_1-w_3+z_1-z_2+t)(-w_1+w_3+z_1-z_2-t)\Big)^{-1}
\\
+ \frac{(w_1-w_2\pm\mu)}{(w_2-w_3+t)}
\\
\times \Big((w_1-w_2+z_1-z_2+t)(-w_1+w_2+z_1-z_2-t)\Big)^{-1}\Big] .
\end{multline}

Now we multiply these formulas by
\be
(-3w_1-t)(w_2-w_3-t)\prod_{l>1} (-w_1+w_l-z_1+z_2-t)(-w_1+w_l+z_1-z_2-t),
\ee
which gives
\be
 (z_1-z_2\pm\mu)(-3w_1-t)(w_2-w_3-t),
\ee
and
\begin{multline}
 (w_1-w_3\pm\mu) 
 (-w_1+w_2-z_1+z_2-t)(-w_1+w_2+z_1-z_2-t)
\\
-\exp[2ir(w_2-w_3)] (w_1-w_2\pm\mu)
 (-w_1+w_3-z_1+z_2-t)(-w_1+w_3+z_1-z_2-t).
\end{multline}

Next, we introduce new variables
\be
x=2w_1+w_2,\ \ \ y=w_1-w_2,\ \ \ u=z_1-z_2.
\ee
Viewing the result as yielding two functions of~$u$, these functions read
\be
\lambda(u):=(u\pm\mu)(-x-y-t)(x-y-t),
\ee
and
\be
\rho(u):=(x\pm\mu)(-y-u-t)(-y+u-t)-\exp[2ir(x-y)](y\pm\mu)(-x-u-t)(-x+u-t).
\ee

We now get
\be
\lambda(u-2t)=\exp[2ir(2u-2t)]\lambda(u),
\ee
and
\be
\rho(u-2t)=\exp[2ir(2u-2t)]\rho(u),
\ee
so that
\be
\tau(u):=\rho(u)/\lambda(u)
\ee
is elliptic. The residue at $u=\mu +t$ is proportional to
\begin{multline}
(x\pm\mu)(-y-\mu-2t)(-y+\mu)-\exp[2ir(x-y)](y\pm\mu)(-x-\mu-2t)(-x+\mu)
\\
=-\exp[2ir(-y-\mu-t)](x\pm\mu)(-y\pm\mu)+\exp[2ir(x-y)+2ir(-x-\mu-t)](y\pm\mu)(-x\pm\mu)
\\
=0,
\end{multline}
whereas the residue at $u=-\mu +t$ is proportional to
\begin{multline}
(x\pm\mu)(-y+\mu-2t)(-y-\mu)-\exp[2ir(x-y)](y\pm\mu)(-x+\mu-2t)(-x-\mu)
\\
=-\exp[2ir(-y+\mu-t)](x\pm\mu)(-y\pm\mu)+\exp[2ir(x-y)+2ir(-x+\mu-t)](y\pm\mu)(-x\pm\mu)
\\
=0,
\end{multline}
so it follows that $\tau(u)$ does not depend on~$u$. Choosing $u=x$, we deduce $\tau(u)=1$. Therefore, the fourth and last residue vanishes as well.
\end{proof}


\section{The $A_3$ identities}

 \subsection{The elliptic  case}
 
Following at first the flow chart of the previous section, we focus on $A_{3,+}(x)$ (given by \eqref{A3sum}--\eqref{A31}), rewritten as
\begin{multline}\label{A3new}
A(x)=\frac{R(x_2-x_3)R(x_3-x_4)R(x_4-x_2)}{R(x_1-x_2-ia_+/2)R(x_1-x_3-ia_+/2)R(x_1-x_4-ia_+/2)}
\\
\times \exp\Big(\frac14 ia_-\big(3\partial_1-\partial_2-\partial_3-\partial_4\big)\Big) + \mathrm{cyclic}.
\end{multline}

As the analog of \eqref{AS2} we  obtain 
\begin{multline}\label{AS3}
\frac{A(v)\cS_3}{\prod_{k,l}G(v_k+w_l-ia_-/4-i(a_++a_-)/4\pm d)}
\\
=\frac{R(v_2-v_3)R(v_3-v_4)R(v_4-v_2)}{\prod_{j=2}^4R(v_1-v_j-t)}
\prod_{l=1}^4R(v_1+w_l-t+c\pm d)+\mathrm{cyclic},
\end{multline}
now with  
\be\label{cA3}
t\equiv ia_+/2,\ \ \ \ c\equiv ia_+/4.
\ee 
 
Like in Section~2, we divide the right-hand side by the $(v\leftrightarrow w)$-invariant product
\be
\prod_{k,l}R(v_k+w_l-t+c\pm d),
\ee
yielding the function
\be\label{cL3}
\cL\equiv \frac{R(v_2-v_3)R(v_3-v_4)R(v_4-v_2)}{\prod_{j=2}^4R(v_1-v_j-t) }
 \prod_{k\ne 1,l}\frac{1}{R(v_k+w_l-t+c\pm d)}+\mathrm{cyclic}.
\ee
The question is now whether $\cL$ equals
\be
\cR\equiv \frac{R(w_2-w_3)R(w_3-w_4)R(w_4-w_2)}{\prod_{j=2}^4R(w_1-w_j-t) }
 \prod_{k,l\ne 1}\frac{1}{R(v_k+w_l-t+c\pm d )}+\mathrm{cyclic}.
\ee

As before, we can show that this is false without the $A_3$ constraint
\be\label{A3}
x_1+x_2+x_3+x_4=0,\ \ \ x=v,w,
\ee
which entails that the elliptic kernel identities do not hold true for unconstrained variables.

\begin{proposition}
Letting $v,w\in\C^4$, the functions $\cL$ and $\cR$ are not equal.
\end{proposition}
\begin{proof}
Viewing  $\cL-\cR$ as a meromorphic function of $v_1$,  this function is a sum of 8 terms that (generically)  have simple poles. The poles at $v_1=v_2+2kt$, $k\in \Z$, only occur  in the first two summands.  Multiplying  by $R(v_1-v_2-t)$, their sum is proportional to
\begin{multline}
\frac{R(v_2-v_3)R(v_4-v_2) }{R(v_1-v_3-t)R(v_1-v_4-t)}
 \prod_{l}\frac{1}{R(v_2+w_l-t+c\pm d)}
 \\
 -\exp[2ir(v_1-v_2)]\frac{R(v_4-v_1)R(v_1-v_3) }{R(v_2-v_3-t)R(v_2-v_4-t)}
 \prod_{l}\frac{1}{R(v_1+w_l-t+c\pm d)} .
\end{multline}
This expression vanishes   for $v_1=v_2$, but for $v_1=v_2+2t$ it becomes proportional to 
\be
 \exp[2ir(2v_2-v_3-v_4)]
 -\exp\Big[2ir\Big(2t-(2v_2-v_3-v_4+2t)+8(v_2+c)+2\sum_{l}w_l \Big)\Big],
\ee
which is nonzero. Hence $\cL$ and $\cR$ are not equal for arbitrary $v,w\in\C^4$, as asserted.
\end{proof}

Henceforth we require the $A_3$ constraint in the explicit form 
\be\label{x4}
x_4\equiv -x_1-x_2-x_3,\ \ \ x=v,w.
\ee
Then the transpositions $x_1\leftrightarrow x_2$, $x_2\leftrightarrow x_3$   and $x_3\leftrightarrow x_4$ generating $S_4$ yield involutions
\be\label{S4r} 
(x_1,x_2,x_3)\mapsto(x_2,x_1,x_3),\ \ (x_1,x_2,x_3)\mapsto (x_1,x_3,x_2),\ \ (x_1,x_2,x_3)\mapsto (x_1,x_2,-x_1-x_2-x_3).
\ee
Also, the A$\De$O~\eqref{A3new} turns into
\begin{multline}\label{A3r} 
A_r(x)=\frac{(x_2-x_3)(x_1+x_2+2x_3)(x_1+2x_2+x_3) }{(x_1-x_2-t)(x_1-x_3-t)(2x_1+x_2+x_3-t)}\exp\big( ia_-\big(3\partial_{1}-\partial_{2} -\partial_{3} \big)/4\big)
 \\
+\frac{(x_1+x_2+2x_3)(2x_1+x_2+x_3)(x_1-x_3) }{(-x_1+x_2-t)(x_2-x_3-t)(x_1+2x_2+x_3-t)}\exp\big( ia_-\big(-\partial_{1}+3\partial_{2} -\partial_{3} \big)/4\big)
\\
+\frac{(2x_1+x_2+x_3)(x_1-x_2)(x_1+2x_2+x_3) }{(x_3-x_1-t)(x_3-x_2-t)(x_1+x_2+2x_3-t) }\exp\big( ia_-\big(-\partial_{1}-\partial_{2} +3\partial_{3} \big)/4\big)
 \\
+\frac{(x_1-x_2)(x_2-x_3)(x_1-x_3) }{(-2x_1-x_2-x_3-t)(-x_1-2x_2-x_3-t)(-x_1-x_2-2x_3-t) } 
\\
\times \exp\big( ia_-\big(-\partial_{1}-\partial_{2} -\partial_{3} \big)/4\big), 
\end{multline}
and the 8-variable function $\cL$~\eqref{cL3} becomes the 6-variable function
\begin{multline}\label{cL3r} 
\cL_r= 
\frac{(v_2-v_3)(v_1+v_2+2v_3)(v_1+2v_2+v_3) }{(v_1-v_2-t)(v_1-v_3-t)(2v_1+v_2+v_3-t)}
\\
\times \Big( \prod_{l}(v_2+w_l-t+c\pm d)(v_3+w_l-t+c\pm d)(-v_1-v_2-v_3+w_l-t+c\pm d)\Big)^{-1}
\\
+\frac{(v_1+v_2+2v_3)(2v_1+v_2+v_3)(v_1-v_3) }{(-v_1+v_2-t)(v_2-v_3-t)(v_1+2v_2+v_3-t)}
\\
\times \Big(\prod_{l}(v_1+w_l-t+c\pm d)(v_3+w_l-t+c\pm d)(-v_1-v_2-v_3+w_l-t+c\pm d)\Big)^{-1}
\\
+\frac{(2v_1+v_2+v_3)(v_1-v_2)(v_1+2v_2+v_3) }{(v_3-v_1-t)(v_3-v_2-t)(v_1+v_2+2v_3-t) }
\\
\times  \Big(\prod_{l}(v_1+w_l-t+c\pm d)(v_2+w_l-t+c\pm d)(-v_1-v_2-v_3+w_l-t+c\pm d) \Big)^{-1}
 \\
+\frac{(v_1-v_2)(v_2-v_3)(v_1-v_3) }{(-2v_1-v_2-v_3-t)(-v_1-2v_2-v_3-t)(-v_1-v_2-2v_3-t) }
\\
\times  \Big(\prod_{l}(v_1+w_l-t+c\pm d)(v_2+w_l-t+c\pm d)(v_3+w_l-t+c\pm d) \Big)^{-1}. 
\end{multline} 
(As before, we have abbreviated $R_+(x)$  to $(x)$.) Using this formula   and the corresponding formula~$\cR_r$ for $\cR$ (obtained from $\cL_r$ by taking $v\leftrightarrow w$),
 the shift $v_1\to v_1+2t$ yields equal multipliers
\be\label{eqmult3}
\exp[2ir(6ia_++6v_2+6v_3)]\exp[2ir(12v_1)],
\ee
for all of the 8 summands. (This follows from straightforward  calculations.)

Just as in Section~2, to conclude that  $\cL_r$ and $\cR_r$ are equal,  we now need only show that the residue sums of $\cL_r$ and $\cR_r$, viewed as  meromorphic and $\pi/r$-periodic functions of~$v_1$, coincide for one of the (generically simple) $v_1$-poles from each pole set $\{p+2kt+l\pi/r \mid k,l\in\Z\}$.   This is straightforward for the $w$-independent poles
\be\label{vpind3}
v_1=v_2,v_3,-2v_2-v_3,-v_2-2v_3,-(v_2+v_3)/2+\omega_j,\ \ j=0,1,2,3. 
\ee
Indeed, $\cR_r$ is regular at the points~\eqref{vpind3}, and like in Lemma~2.4, residue cancellation for the 2 pertinent summands of $\cL_r$ at the 8 poles~\eqref{vpind3} is  readily verified.    The remaining problem is therefore to demonstrate equality of residue sums
at the pole
\be\label{vpdep3}
v_1=-w_1-c+ d,\ \ c=ia_+/4.
\ee
(To see that this suffices, note that we can invoke $S_4$ symmetry in~$w$ and evenness in~$d$ for the 7 remaining $w$-dependent poles.)   

We have no doubt this equality holds true, but as already mentioned in the Introduction, we have not found a complete proof. For the hyperbolic specialization, however, we prove the kernel identities in the next subsection. This entails the equality of the hyperbolic version of the above functions $\cL_r$ and $\cR_r$.

 \subsection{The hyperbolic  case}

Here we still have~\eqref{kf3} and~\eqref{A3sum}, now with $G$ the hyperbolic gamma function (cf.~Appendix~A), and \eqref{A31} replaced by
\be\label{A31h}
A_{3,\de}^{(1)}(x):= \frac{c_{\de}(x_2-x_3)c_{\de}(x_3-x_4)c_{\de}(x_4-x_2)}{\prod_{j=2}^4s_{\de}(x_1-x_j) }\exp\Big(\frac14 ia_{-\de}\big(3\partial_1-\partial_2-\partial_3-\partial_4\big)\Big).
\ee
Here and in the sequel we use abbreviations
\be\label{hnot}
c_{\de}(x):=\cosh(\pi x/a_{\de}),\ \ \ \ s_{\de}(x):=\sinh(\pi x/a_{\de}),\ \ \ \ e_{\de}(x):=\exp(\pi x/a_{\de}).
\ee
  Focussing once again on $A_{3,+}$, the counterpart of \eqref{AS2} becomes
\begin{multline} \label{AvS}
\frac{A_{3,+}(v)\cS_3}{\prod_{k,l=1}^4G(v_k+w_l-ia_-/4+i(a_++a_-)/4\pm d)}=\frac{c_+(v_2-v_3)c_+(v_3-v_4)c_+(v_4-v_2)}{s_+(v_1-v_2)s_+(v_1-v_3) s_+(v_1-v_4)  }
\\
\times \prod_{l=1}^42c_+(v_1+w_l-ia_+/4\pm d)+\mathrm{cyclic}.
\end{multline}

Suppressing the subscript + from now on, we divide the right-hand side by
\be
16\prod_{k,l}s(v_k+w_l+c\pm d),\ \ \ c=ia_+/4,
\ee
to get the function
\be\label{cLh}
\cL\equiv \frac{c(v_2-v_3)c(v_3-v_4)c(v_4-v_2)}{\prod_{j>1}s(v_1-v_j) }
 \prod_{k\ne 1,l}\frac{1}{s(v_k+w_l+c\pm d)}+\mathrm{cyclic}
 =: \sum_{j=1}^4 \cL_j.
\ee
The question is now whether $\cL$ equals
\be\label{cRh}
\cR\equiv \frac{c(w_2-w_3)c(w_3-w_4)c(w_4-w_2)}{\prod_{j>1}s(w_1-w_j)  } \prod_{k,l\ne 1}\frac{1}{s(v_k+w_l+c\pm d )}+\mathrm{cyclic}
 =: \sum_{j=1}^4 \cR_j.
\ee

Assuming at first that the 8 variables $v_j,w_j$ are independent, there is an illuminating way to see that this is not true. (Of course, we have already shown in Prop.~3.1 that  $\cL$ and $\cR$ are unequal for the elliptic unrestricted  case, but this involves an argument that has no hyperbolic analog. A priori, equality might therefore hold for the hyperbolic specialization.) 

\begin{proposition}
Letting $v,w\in\C^4$, the functions $\cL$ and $\cR$ are not equal.
\end{proposition}
\begin{proof}  
We multiply $\cL$ and $\cR$ by 
\be
\prod_{j=1}^3 s(v_j+w_j+c-d),
\ee
and then set
\be\label{vj}
v_j=-w_j-c+d,\ \ \ j=1,2,3.
\ee
Then we only get a nonzero contribution from $\cL_4$ and $\cR_4$, viz.,
\begin{multline}\label{L4}
\frac{c(w_1-w_2)c(w_2-w_3)c(w_3-w_1)}{\prod_{j=1}^3 s(v_4+w_j+c-d) }
\\
\times \prod_{k\ne l,k,l=1}^3\frac{1}{s(w_l-w_k+d\pm d )}\prod_{k=1}^3
\frac{1}{s(w_4-w_k+d \pm d )s(2d)},
\end{multline}
and
\begin{multline}\label{R4}
\frac{c(w_1-w_2)c(w_2-w_3)c(w_3-w_1) }{\prod_{j=1}^3 s(w_4-w_j)  }
\\
\times \prod_{k\ne l,k,l=1}^3\frac{1}{s(w_l-w_k+d\pm d )}\prod_{k=1}^3
\frac{1}{s(v_4+w_k+c \pm  d )s(2d)}.
\end{multline}
Since~\eqref{R4} has poles for $v_4=-w_k-c-d$, $k=1,2,3$, which are not shared by~\eqref{L4} (for generic $d$), it follows that $\cL$ and $\cR$ are not equal, as asserted.
\end{proof}

By contrast to the state of affairs in this proof, requiring the $A_3$ constraint~\eqref{A3} we obtain from~\eqref{vj} 
\be
v_4=-v_1-v_2-v_3=w_1+w_2+w_3+3c -3d=-w_4+3c -3d.
\ee
Thus, canceling the equal factors occurring in~\eqref{L4} and~\eqref{R4}, we are left with
\be\label{L4s}
\frac{1}{\prod_{j=1}^3s(w_j-w_4+4c-4d) }
\prod_{k=1}^3\frac{1}{s(w_4-w_k+d \pm d )},
\ee
and
\be\label{R4s}
\frac{1}{\prod_{j=1}^3s(w_4-w_j) }\prod_{k=1}^3
\frac{1}{s(w_k-w_4+4c -3d \pm  d )}.
\ee
When we now use 
\be
s(x+4c)=-s(x),
\ee
we deduce that \eqref{L4s} and \eqref{R4s} are equal.
Of course, this does not yet prove that the restricted $\cL$ and $\cR$ are equal, but it is a quite suggestive test. 

At this stage we abandon consideration of the restricted functions $\cL_r$ and $\cR_r$, so as to embark on another strategy that avoids the study of $v_1$-poles depending on~$w$. (As mentioned before, the latter poles give rise to a snag that may be surmountable, but we wish to avoid that climb... To be sure, it does follow from the sequel that $\cL_r$ and $\cR_r$  are indeed equal in the hyperbolic case at hand.) 

To begin with, let us reconsider~\eqref{AvS}. Our problem is to show that its rhs
 \be\label{vwsymm}
 \frac{c_+(v_2-v_3)c_+(v_3-v_4)c_+(v_4-v_2)}{s_+(v_1-v_2)s_+(v_1-v_3)s_+(v_1-v_4) }
 \prod_{l=1}^4s_+(v_1+w_l+ia_+/4\pm d)+\mathrm{cyclic}
 \ee
 is ($v\leftrightarrow w$)-invariant when we require the $A_3$ sum constraint~\eqref{A3}. Next we write
 \be\label{sadd}
 2s_+(v_k+w_l+ia_+/4\pm d)=c_+(2v_k+2w_l+ia_+/2)-c_+(2d)=is_+(2v_k+2w_l)-c_+(2d).
 \ee
 We now introduce
 \be
 b:= ic_+(2d),
 \ee
 noting that our problem then amounts to proving ($v\leftrightarrow w$)-invariance of
  \be
 \frac{c(v_2-v_3)c(v_3-v_4)c(v_4-v_2)}{s(v_1-v_2)s(v_1-v_3)s(v_1-v_4)}
 \prod_{l=1}^4[s(2v_1+2w_l)+b] +\mathrm{cyclic},
 \ee
 with \eqref{A3} in force. (Here and from now on, we again suppress the subscript +.) 
 
 Dividing by the manifestly $(v\leftrightarrow w)$-invariant product
 \be
 8\prod_{1\le j<k\le 4}c(v_j-v_k)c(w_j-w_k),
 \ee
we next define
\be
F_L:= \frac{1}{\prod_{1\le j<k\le 4}c(w_j-w_k)}
\Big(  \frac 
{ \prod_{l=1}^4[s(2v_1+2w_l)+b]}{s(2v_1-2v_2)s(2v_1-2v_3)s(2v_1-2v_4)} +\mathrm{cyclic}\Big),
\ee
and 
\be
F_R:= \frac{1}{\prod_{1\le j<k\le 4}c(v_j-v_k)}
\Big(  \frac 
{ \prod_{k=1}^4[s(2v_k+2w_1)+b]}{s(2w_1-2w_2)s(2w_1-2w_3)s(2w_1-2w_4)} +\mathrm{cyclic}\Big).
\ee
We are now prepared to state the main result of this section.

\begin{theorem}
Under the constraint
\be\label{A3cons}
\sum_{j=1}^4v_j=\sum_{j=1}^4w_j=0,
\ee
the $A_3$ hyperbolic kernel identities~\eqref{kid3} hold true.
\end{theorem}
\begin{proof}
We have already reduced the proof of this assertion to showing that the functions $F_L$ and $F_R$ become equal when~\eqref{A3cons} is imposed.
To prove equality, let us view $F_L$ as a function of $v_1$. It is $ia_+$-periodic, both for the case of unrestricted~$v$ and for the $A_3$-case $\sum_j v_j=0$. In the former case it has no poles for $v_1=v_2$, and it is easily checked it has no poles there in the latter case either. (Evidently, the same is true for $F_R$.) However, for
\be\label{v1v2}
v_1=v_2+ia_+/2,
\ee
we do get (generically)  simple poles from the first two summands of~$F_L$, whose residues do not cancel.

Viewing $F_R$ as a function of $v_1$, it is $ia_+$-antiperiodic for unrestricted~$v$ (entailing once more $F_L\ne F_R$ in that case). On the other hand, imposing  the $A_3$-restriction~\eqref{A3cons} in the explicit form~\eqref{x4} from now on, we have
\be\label{cprod}
\prod_{1\le j<k\le 4}c(v_j-v_k)=c(v_1-v_2)c(v_1-v_3)c(2v_1+v_2+v_3)c(v_2-v_3)c(v_1+2v_2+v_3)c(v_1+v_2+2v_3).
\ee
Hence the restricted function $F_{R,r}$ that results from requiring~\eqref{x4}  is $ia_+$-periodic in~$v_1$, just as the restricted function~$F_{L,r}$.

The second and third summand of   $F_{L,r}$ vanish for $|\re v_1|\to\infty$, whereas from the first and fourth we get for $\re v_1\to\infty$ proportionality to
\begin{multline}
\frac{\prod_le(2v_1+2w_l)}{e(2v_1-2v_2)e(2v_1-2v_3)e(4v_1+2v_2+2v_3)}
\\
-
\frac{\prod_le(2v_1+2v_2+2v_3-2w_l)}{e(4v_1+2v_2+2v_3)e(2v_1+4v_2+2v_3)e(2v_1+2v_2+4v_3)}=\prod_le(2w_l)-\prod_le(-2w_l).
\end{multline}
Due to the $w$-restriction, this limit vanishes. Likewise, we infer $F_{L,r}\to 0$ for $\re v_1\to -\infty$. Clearly, we also have $F_{R,r}\to 0$ for $\re v_1\to \pm\infty$. For equality of $F_{L,r}$ and $F_{R,r}$ it therefore suffices to show that the residues at the  $v_1$-poles of $F_{L,r}$ and $F_{R,r}$ are equal.

Just as in the elliptic case,  both functions are regular for
\be
v_1=v_2,v_3,-2v_2-v_3, -v_2-2v_3,-(v_2+v_3)/2,-(v_2+v_3)/2+ia_+/2.
\ee
  (The summands of $F_{R,r}$ are regular at these points, and it is easy to check residue cancellation for the two relevant summands of~$F_{L,r}$.) 

It involves a lot more work to show equality of residues for
\be
v_1-ia_+/2=v_2,v_3,-2v_2-v_3, -v_2-2v_3,-(v_2+v_3)/2-ia_+/4,-(v_2+v_3)/2-3ia_+/4.
\ee
To begin with, we observe that $F_{L,r}$ and $F_{R,r}$ are invariant under swapping $v_2$ and $v_3$. Thus we need only consider the residues at the first, third, fifth and sixth location.
To study the residues at the first location, we multiply by $c(v_1-v_2)$ and then put $v_1=v_2+ia_+/2$. For $F_{L,r}$ this yields 
\be\label{r1}
r_L 
=\frac{\prod_{l=1}^4[s(2v_2+2w_l)+b]- \prod_{l=1}^4[s(2v_2+2w_l)-b]}{2is(2v_2-2v_3)s(6v_2+2v_3)\prod_{1\le j<k\le 4}c(w_j-w_k)},
\ee
and for $F_{R,r}$ we obtain, using~\eqref{cprod},
\begin{multline}\label{r2}
r_R= \frac{-4i}{s(2v_2-2v_3)s(6v_2+2v_3)s(2v_2+2v_3)}
\\
\times\Big(  \frac 
{ [b^2-s(2w_1+2v_2)^2]  [b+s(2w_1+2v_3)] [b-s(2w_1-4v_2-2v_3)]}{s(2w_1-2w_2)s(2w_1-2w_3)s(2w_1-2w_4)} +\mathrm{cyclic}\Big).
\end{multline}

As it turns out, we need only prove equality of $r_L$ and $r_R$. This is because long, but straightforward calculations show that when we proceed in the same way for the third location $v_1=-2v_2-v_3+ia_+/2$ (that is, multiply by $c(v_1+2v_2+v_3)$, etc.), we obtain $-r_L$ and $-r_R$, whereas for  the fifth and sixth location we get upon multiplication by $c(2v_1+v_2+v_3)$, setting $v_1=-(v_2+v_3)/2+i\de a_+/4$, $\de=+,-$, and then substituting
\be
(v_2,v_3)\mapsto (-2v_2-v_3+ia_+/2,v_3),
\ee
once more the   quantities $-r_L$ and $-r_R$.
(These unexpected coincidences can possibly be understood as a footprint left by the original $S_4$ symmetry, but we do not see an argument from which this follows.)

Comparing $r_L$ and $r_R$, we infer that our remaining task is to prove equality of
\be\label{rl}
 \prod_{l=1}^4[b+s(2v_2+2w_l)]- \prod_{l=1}^4[b-s(2v_2+2w_l)]
\ee
and
\begin{multline}\label{rr}
\frac{1}{s(2v_2+2v_3)}\Big(  \frac 
{ c(w_2-w_3)c(w_3-w_4)c(w_4-w_2) }{s(w_1-w_2)s(w_1-w_3)s(w_1-w_4)}
\\
\times [b^2-s(2w_1+2v_2)^2]  [b+s(2w_1+2v_3)] [b-s(2w_1-4v_2-2v_3)] +\mathrm{cyclic}\Big)
\\
=:\frac{1}{s(2v_2+2v_3)}S_R(v_2,v_3).
\end{multline}
In particular, we should show that~\eqref{rr} does not depend on $v_3$, just as~\eqref{rl}. (Note that at face value this $v_3$-independence seems hard to believe.) 

To prove the latter property, consider the ratio
\be\label{SRsp}
\frac{S_R(v_2,-v_2)}{8\prod_{1\le j<k\le 4}c(w_j-w_k)}=  \frac 
{ [b^2-s(2w_1+2v_2)^2] [b^2-s(2w_1-2v_2)^2]  }{s(2w_1-2w_2)s(2w_1-2w_3)s(2w_1-2w_4)}
 +\mathrm{cyclic}
\ee
as a function of $w_1$. It is $ia_+/2$-periodic and has simple poles at
\be\label{w1p}
w_1=w_2,w_3,-(w_2+w_3)/2,-(w_2+w_3)/2+ia_+/4,
\ee
in a period strip. It is readily checked that the residue sums at these poles vanish.
For $\re w_1\to\pm\infty$ the second and third summand on the rhs of~\eqref{SRsp} vanish. For $\re w_1\to\infty$ we get from the first and fourth summand
\begin{multline}
\frac{e(4w_1+4v_2)e(4w_1-4v_2)}{e(2w_1-2w_2)e(2w_1-2w_3)
e(4w_1+2w_2+2w_3)}
\\
-
\frac{e(4w_1+4w_2+4w_3-4v_2)e(4w_1+4w_2+4w_3+4v_2)}{e(4w_1+2w_2+2w_3)e(2w_1+4w_2+2w_3)
e(2w_1+2w_2+4w_3)}=0.
\end{multline}
Likewise, the sum vanishes for $\re w_1\to -\infty$. Hence it follows that the ratio on the lhs of~\eqref{SRsp} vanishes.

As a consequence, $S_R(v_2,-v_2)$ vanishes. In the same way, it follows that $S_R(v_2,-v_2+ia_+/2)$ vanishes. Hence $S_R(v_2,v_3)/s(2v_2+2v_3)$, viewed as an $ia_+$-periodic function of $v_3$, has no poles. Moreover, for $\re v_3\to \infty$ we have, using the definition~\eqref{rr} of~$S_R(v_2,v_3)$,
\begin{multline}
\frac{4S_R(v_2,v_3)}{8\prod_{1\le j<k\le 4}c(w_j-w_k)}=\frac 
{  [b^2-s(2w_1+2v_2)^2]e(2v_3+2w_1)e(2v_3-2w_1+4v_2)  }{s(2w_1-2w_2)s(2w_1-2w_3)s(2w_1-2w_4)}
\\
  +\mathrm{cyclic}+O(e(2v_3)) 
  \\
=e(4v_3+4v_2)\Big( \frac 
{  [b^2-s(2w_1+2v_2)^2]   }{s(2w_1-2w_2)s(2w_1-2w_3)s(2w_1-2w_4)}  +\mathrm{cyclic}\Big)+O(e(2v_3)).
\end{multline}
Just as for the rhs of~\eqref{SRsp}, it can be verified that the function in brackets vanishes. For $\re v_3\to -\infty$ we get a similar result. Hence we obtain
\be
S_R(v_2,v_3)/s(2v_2+2v_3)=O(1),\ \ \ \re v_3\to \pm\infty,
\ee
and since this ratio is pole-free, it follows that it equals its limit for $v_3\to\infty$. Therefore, \eqref{rr} equals
\begin{multline}
 e(-2v_2) \frac 
{ c(w_2-w_3)c(w_3-w_4)c(w_4-w_2) }{s(w_1-w_2)s(w_1-w_3)s(w_1-w_4)}
\\
\times [b^2-s(2w_1+2v_2)^2]  b(e(-2w_1+4v_2)+e(2w_1))+\mathrm{cyclic}
\\
=2b\frac 
{ c(w_2-w_3)c(w_3-w_4)c(w_4-w_2) }{s(w_1-w_2)s(w_1-w_3)s(w_1-w_4)}
 [b^2-s(2w_1+2v_2)^2]c(2w_1-2v_2)+\mathrm{cyclic}. 
\end{multline}

Comparing this to~\eqref{rl}, we see that we can complete the proof of the theorem by invoking the identities~\eqref{id1} and~\eqref{id2} in the next lemma.
\end{proof} 

\begin{lemma} 
Letting $x,w_1,w_2,w_3,w_4\in\C$ and $\sum_{j=1}^4w_j=0$, we have identities
\be\label{id0}
\frac 
{ c(w_2-w_3)c(w_3-w_4)c(w_4-w_2) }{s(w_1-w_2)s(w_1-w_3)s(w_1-w_4)}
+\mathrm{cyclic}=0,
\ee
\be\label{id1}
\frac 
{ c(w_2-w_3)c(w_3-w_4)c(w_4-w_2) }{s(w_1-w_2)s(w_1-w_3)s(w_1-w_4)}
c(x-2w_1)+\mathrm{cyclic}=\sum_{l=1}^4s(x+2w_l),
\ee
\begin{multline}\label{id2}
\frac 
{ c(w_2-w_3)c(w_3-w_4)c(w_4-w_2) }{s(w_1-w_2)s(w_1-w_3)s(w_1-w_4)}
c(x-2w_1)s(x+2w_1)^2+\mathrm{cyclic}
\\
=-\sum_{1\le l_1<l_2<l_3\le 4}s(x+2w_{l_1})s(x+2w_{l_2})s(x+2w_{l_3}).
\end{multline}
 \end{lemma}

\begin{proof} 
 In order to prove \eqref{id0}, we begin by noting that its lhs, viewed as a function  for unconstrained $w\in\C^4$, is holomorphic on $\C^4$. Indeed, it has no poles on the hyperplanes $w_j=w_k+ila_+$, $l\in\Z$, since residues cancel in pairs. (We get simple poles for generic values of the remaining variables, and need not consider hyperplane intersections, since they have codimension 2.)  Restricting attention to the $A_3$ choice
\be
w_4 \equiv -w_1-w_2-w_3,
\ee
and viewing the lhs from now on as a function of $w_1$, we therefore obtain an entire function. (Alternatively, residue cancellation at the $w_1$-poles can be checked directly.) 

Denoting the coefficient ratios by $R_j(w)$, $j=1,2,3,4$,  we now have for $\re w_1\to  \infty$:
\begin{multline}\label{R1}
R_1=\frac{ c(w_2-w_3)c(w_1+w_2+2w_3)c(w_1+2w_2+w_3) }{s(w_1-w_2)s(w_1-w_3)s(2w_1+w_2+w_3)}
\\
=2c(w_2-w_3)\Big(  e( - 2w_1+3w_2+3w_3)+O(e(-4w_1) )  \Big), 
\end{multline}
\begin{multline}\label{R2}
R_2= \frac 
{ c(w_1-w_3)c(2w_1+w_2+w_3 )c(w_1+w_2+2w_3) }{s(w_2-w_1)s(w_2-w_3)s(w_1+2w_2+w_3)}
=\frac{-1}{2s(w_2-w_3)}\Big(  [e( 2w_1+w_2+w_3)
\\
 +e(w_2+3w_3)+e(-w_2-3w_3)+e(3w_2+w_3)+e(-3w_2-w_3)    ]+O(e(-2w_1) )  \Big), 
\end{multline}
\begin{multline}\label{R3}
R_3= \frac 
{ c(w_1-w_2)c(2w_1+w_2+w_3 )c(w_1+2w_2+w_3) }{s(w_3-w_1)s(w_3-w_2)s(w_1+w_2+2w_3)}
\\
=\frac{-1}{2s(w_3-w_2)}\Big(  [e( 2w_1+w_2+w_3) +2c(w_2+3w_3)+2c(3w_2+w_3)    ]+O(e(-2w_1))  \Big), 
\end{multline}
\begin{multline}
R_4=-\frac{ c(w_2-w_3)c(w_1-w_2)c(w_1-w_3) }{s(2w_1+w_2+w_3)s(w_1+2w_2+w_3)s(w_1+w_2+2w_3)}
\\
=-2c(w_2-w_3)\Big(  e(- 2w_1-5w_2-5w_3)+O(e(-4w_1))  \Big). 
\end{multline}
As a result, we obtain
\be
\sum_{m=1}^4 R_m= O(e(-2w_1)),
\ee
so that the lhs of~\eqref{id0} vanishes for $\re w_1\to \infty$. The lhs changes sign when we take $w\to -w$, whence we deduce it also vanishes for $\re w_1\to -\infty$. Since it is $ia_+$-periodic and pole-free, it must vanish identically, yielding  the first identity~\eqref{id0}.

In order to prove \eqref{id1}, we proceed along the same lines. Once more, it follows that we are dealing with functions that are entire in the variable~$w_1$. On the lhs we now obtain for $\re w_1\to  \infty$:
\be
R_1c(x-2w_1)=c(w_2-w_3) e(- x+3w_2+3w_3)+O(e(-2w_1) ),
\ee 
\begin{multline}
R_2c(x-2w_2)+R_3c(x-2w_3)=\frac{1}{2s(w_2-w_3)} [-c(x-2w_2)+c(x-2w_3)]
\\
 \times [e( 2w_1+w_2+w_3) +2c(w_2+3w_3)+2c(3w_2+w_3)    ]
+O(e(-2w_1) )
\\
=s(x-w_2-w_3)[e( 2w_1+w_2+w_3) +2c(w_2+3w_3)+2c(3w_2+w_3)    ]
+O(e(-2w_1) ), 
\end{multline} 
and
\be
R_4c(x-2w_4)=-c(w_2-w_3) e ( x-3w_2-3w_3)+O(e(-2w_1)) .
\ee 

Summing  these three terms, we see that the lhs of~\eqref{id1} yields
\begin{multline}
s(x-w_2-w_3)[2c( 2w_1+w_2+w_3) +2c(w_2+3w_3)+2c(3w_2+w_3)    ]
\\
-2c(w_2-w_3)s(x-3w_2-3w_3)+O(e(-2w_1) )
\\
=s(x+2w_1)+s(x-2w_1-2w_2-2w_3) +s(x+2w_3)+s(x-2w_2-4w_3)+s(x+2w_2)
\\
+s(x-4w_2-2w_3)
-s(x-2w_2-4w_3)-s(x-4w_2-2w_3)+O(e(-2w_1) )
\\
=\sum_{l=1}^4s(x+2w_l)+O(e(-2w_1) ).
\end{multline}
As a result, we get limit 0 for $\re w_1\to\infty$ when we subtract~$\sum_{l}s(x+2w_l)$ from the lhs of~\eqref{id1}. Now when we flip the sign of $x$ and $w$, both lhs and rhs of~\eqref{id1} change sign. Therefore the difference of lhs and rhs also vanishes for $\re w_1\to -\infty$. Taking entireness and periodicity into account, \eqref{id1} follows. 

To prove the identity~\eqref{id2}, we multiply by
\be
16\prod_{1\le j<k\le 4}s(w_j-w_k)\Big/\prod_{l=1}^4 s(x+2w_l),
\ee
to get the equivalent identity
\begin{multline}\label{id3}
\frac{2c(x-2w_1)s(x+2w_1)}{\prod_{l>1}s(x+2w_l)}s(2w_2-2w_3)s(2w_3-2w_4)s(2w_2-2w_4)+\mathrm{cyclic}
\\
=-16\prod_{j<k}s(w_j-w_k)\sum_{l=1}^4\frac{1}{s(x+2w_l)}.
\end{multline}
Both lhs and rhs are antisymmetric under permutations of $w_1,w_2,w_3,w_4$, vanish for $\re x\to\pm\infty$, and are $ia_+$-antiperiodic in~$x$. Multiplying by $s(x+2w_4)$ and then putting $x=-2w_4$, we obtain on the lhs 
\begin{multline}\label{lhs}
2c(2w_4+2w_1)s(2w_1-2w_4)s(2w_2-2w_3)-2c(2w_4+2w_2)s(2w_2-2w_4)s(2w_1-2w_3)
\\
+2c(2w_4+2w_3)s(2w_3-2w_4)s(2w_1-2w_2)
\\
=[s(4w_1)-s(4w_4)]s(2w_2-2w_3) +[s(4w_2)-s(4w_4)]s(2w_3-2w_1)
\\
+[s(4w_3)-s(4w_4)]s(2w_1-2w_2).
\end{multline}
Clearly, to prove~\eqref{id3} it suffices to show that this expression equals $-16\prod_{j<k} s(w_j-w_k)$.

To this end we expand the $s$-product by using
\be\label{ss}
2s(a)s(b)=c(a+b)-c(a-b).
\ee
This yields
\be
\prod_{1\le j<k\le 3}2s(w_j-w_k)=-2[s(2w_1-2w_2) +s(2w_2-2w_3)+s(2w_3-2w_1)],
\ee 
and
\begin{multline}
\prod_{l=1}^3 2s(w_l-w_4)=2s(w_1+w_2+w_3-3w_4)
\\
-2s(w_1+w_2-w_3-w_4)-2s(w_1-w_2+w_3-w_4)-2s(-w_1+w_2+w_3-w_4).
\end{multline}
Using next $w_4=-w_1-w_2-w_3$, we get
\be
\prod_{l=1}^3 2s(w_l-w_4)=-2[s(4w_4)+s(2w_1+2w_2)+s(2w_2+2w_3)+s(2w_3+2w_1)].
\ee
From this we deduce
\begin{multline}
-16\prod_{1\le j<k\le 4}s(w_j-w_k) = -\big[ s(2w_1-2w_2)+s(2w_2-2w_3)+s(2w_3-2w_1)\big]
\\
\times \big[ s(4w_4)+s(2w_1+2w_2)+s(2w_2+2w_3)+s(2w_3+2w_1)\big].
\end{multline}
Comparing this to~\eqref{lhs}, we infer that it remains to show the identity
\begin{multline}
- \big[ s(2w_1-2w_2)+s(2w_2-2w_3)+s(2w_3-2w_1)\big]
\\
\times \big[ s(2w_1+2w_2)+s(2w_2+2w_3)+s(2w_3+2w_1)\big]
\\
=s(4w_1)s(2w_2-2w_3)+s(4w_2)s(2w_3-2w_1)+s(4w_3)s(2w_1-2w_2).
\end{multline}
Expanding lhs and rhs by using once more~\eqref{ss}, we see that among the 18 cosh-terms on the lhs, the 6 terms involving $c(4w_j)$ and the 6 terms involving $c(2w_j+2w_k)$ with $j<k$ cancel in pairs. The remaining 6 terms yield the 6 terms on the rhs. 
 \end{proof}

\section{Further developments}

\subsection{The $A_2$ case}

Consider the A$\De$O family $A_{2,\de}(\mu;x)$ \eqref{A2sum}. It is not obvious, but true that all of these operators commute, viewed as endomorphisms of the space of meromorphic functions of~$x$. We proceed to prove this. 

\begin{proposition}
We have commutation relations
\be\label{Acom}
[A_{\de}(\mu;x),A_{\de'}(\mu';x)]=0,
\ee
where $\de,\de'\in \{+,-\}$ and $\mu,\mu'\in\C$.
\end{proposition}
\begin{proof}
For $\de\ne\de'$, this is easily checked from the definition~\eqref{A2sum} by using the A$\De$E~\eqref{Rade} satisfied by $R_{\de}(z)$. Taking next $\de=\de'$, we may just as well show commutativity for the A$\De$Os obtained upon multiplication by $\exp[ia_{-\de}(\partial_1+\partial_2+\partial_3)/3]$. This amounts to showing that the coefficients of the shift monomials $\exp[ia_{-\de}(\partial_k+\partial_l)]$ in the commutator all vanish. This is clear for $k=l$, so by $S_3$-invariance we need only show vanishing for $k=1$ and $l=2$, and we may also take $\de=+$.

A short calculation shows that the vanishing of the pertinent coefficient is equivalent to the identity
\begin{multline}\label{Razaid}
\frac{R_+(x_2-x_3\pm\mu)R_+(x_1-x_3+ia_-\pm\mu')}{R_+(x_1-x_2-ia_+/2)R_+(x_2-x_1-ia_--ia_+/2)}
\\
+\frac{R_+(x_1-x_3\pm\mu)R_+(x_2-x_3+ia_-\pm\mu')}{R_+(x_2-x_1-ia_+/2)R_+(x_1-x_2-ia_--ia_+/2)}=(\mu\leftrightarrow \mu').
\end{multline}
Setting
\be
u:=x_2-x_3,\ \ v:=x_1-x_3,\ \ w:=ia_-,\ \ t:=ia_+/2,
\ee
and dropping the subscript, this becomes
\be\label{Raid}
\frac{R(u\pm\mu)R(v+w\pm\mu')}{R(u-v+t )R(v-u+w+t)}
+\frac{R(v\pm\mu)R(u+w\pm\mu')}{R(v-u+t)R(u-v+w+t)}=(\mu\leftrightarrow \mu').
\ee

In order to prove~\eqref{Raid}, we divide the lhs and rhs by $R(u+v-t)R(u+v+w-t)$. The point is that the resulting 4 functions are then elliptic functions of $w$ (with periods $\pi/r$, $ia_+$). The residues of lhs and rhs at the (generically) simple poles
\be
w=u-v,v-u, -u-v,
\ee
are readily verified to be equal. Therefore, the difference is constant in~$w$. Choosing $w=0$, we deduce that the difference vanishes, so that~\eqref{Acom} follows.
\end{proof}

It is clear from the proof that~\eqref{Acom} still holds when we view the operators as endomorphisms of the space of meromorphic functions of $x$ satisfying the $A_2$ restriction $\sum_jx_j=0$. Razamat  reduced the  commutativity of his operators to the theta-function identity Eq.\,(2.18) in~\cite{Ra18}, which is substantially equivalent to~\eqref{Razaid}.

Letting $x\in\R^3$, it follows from~\eqref{defR} that the factors $R_{\de}(x_j-x_k\pm\mu)$ take values in $[0,\infty)$  when we choose $\mu$ real or purely imaginary. We shall do so from now on.

Next we note that the shifts are formally self-adjoint (henceforth s.\,a.) with respect to Lebesgue measure on $\R^3$. Viewed as operators on meromorphic functions, they commute with the numerator functions, but they do not commute with the denominator functions. Moreover, for $x\in \R^3$ the latter are not real-valued. Thus the A$\De$Os are not  formally s.\,a. with respect to Lebesgue measure. 

However, there is a simple choice for a weight function $W(x)$ so that the A$\De$Os become formally s.\,a.~with respect to the measure $W(x)dx$, as already pointed out by Razamat~\cite{Ra18}. With our conventions this weight function reads
\be\label{W}
W(x):=\prod_{1\le j<k\le 3}G(\pm (x_j-x_k)+ia).
\ee
Using the reflection equation~\eqref{refle} and A$\De$E~\eqref{Geratio}, this yields a simpler representation
\be\label{Ws}
W(x)=(p_+p_-)^3\prod_{1\le j<k\le 3}s_+(x_j-x_k)s_-(x_j-x_k).
\ee

Choosing from now on
\be\label{xord}
-\pi/2r\le x_3<x_2<x_1<\pi/2r,
\ee
the formal self-adjointness property just announced can be deduced from the similarity transformation to new A$\De$Os
\be\label{Hde}
H_{2,\de}(\mu;x):=W(x)^{1/2}A_{2,\de}(\mu;x)W(x)^{-1/2},\ \ \ \de=+,-,
\ee
where the positive square root of $W(x)$ is understood. (Note~\eqref{xord} implies $W(x)$ is positive.) In fact, from an explicit formula for these Hamiltonians we can obtain a stronger property, as shown next.

\begin{proposition}
With the ordering~\eqref{xord} in effect and $\mu\in \R\cup i\R$, the A$\De$Os $H_{2,\de}(\mu;x)$ are formally positive with respect to Lebesgue measure.
\end{proposition}
\begin{proof}
We claim that the above Hamiltonians can be written as
\begin{multline}\label{H2sum} 
H_{2,\de}(\mu;x)=\frac{\exp(-ra_{-\de})}{p_{\de}^2}\Big[ R_{\de}(x_2-x_3\pm\mu)V_{1,\de}(x)^{1/2}\exp\Big(\frac13 ia_{-\de}\big(2\partial_{x_1}-\partial_{x_2}-\partial_{x_3}\big)\Big)V_{1,\de}(x)^{1/2}
\\
+ R_{\de}(x_3-x_1\pm\mu)V_{2,\de}(x)^{1/2}\exp\Big(\frac13 ia_{-\de}\big(2\partial_{x_2}-\partial_{x_3}-\partial_{x_1}\big)\Big)V_{2,\de}(x)^{1/2}
\\
+ R_{\de}(x_1-x_2\pm\mu)V_{3,\de}(x)^{1/2}\exp\Big(\frac13 ia_{-\de}\big(2\partial_{x_3}-\partial_{x_1}-\partial_{x_2}\big)\Big) V_{3,\de}(x)^{1/2}\Big] ,
\end{multline}
where we have introduced coefficients
\be\label{V1}
V_{1,\de}(x):=1/s_{\de}(x_1-x_2)s_{\de}(x_1-x_3),
\ee
\be
 V_{2,\de}(x):=1/s_{\de}(x_1-x_2)s_{\de}(x_2-x_3),
 \ee
 \be\label{V3}
 V_{3,\de}(x):=1/s_{\de}(x_1-x_3)s_{\de}(x_2-x_3),
\ee
and positive square roots are understood. (Note that these coefficients are positive in view of our assumption~\eqref{xord}.) 

Taking this explicit representation for granted, formal positivity is manifest from the $\mu$-assumption ensuring $R_{\de}(x_j-x_k\pm \mu)\in [0,\infty)$ and the formal positivity of the shifts. Its verification involves crucial and tricky signs, so we add a little more detail. Letting $y\in (0,\pi/r)$, the key point is to use the relation~\eqref{sR} between $s_{\de}(y)$ and $R_{\de}(y)$ to conclude
\be\label{sRy}
\frac{\sqrt{s_{\de}(y)}}{R_{\de}(sy-ia_{\de}/2)}=\frac{si\exp(-siry)}{p_{\de}}\frac{1}{\sqrt{s_{\de}(y)}},\ \ \ \ \ \ \ \ \ \ \ s=+,-,
\ee
and combine this with
\be\label{ssy}
\exp(sia_{-\de}d/dy)\frac{1}{\sqrt{s_{-\de}(y)}}=-si\exp(-ra_{-\de}/2)\exp(siry)\frac{1}{\sqrt{s_{-\de}(y)}}\exp(sia_{-\de}d/dy),
\ee
where we used $y\in(0,\pi/r)$ and  the A$\De$E~\eqref{Rade} satisfied by $s_{-\de}$.
\end{proof}

With the $A_2$ sum constraint in force from now on, the Hamiltonians \eqref{H2sum}
satisfy the kernel identities
\be\label{kid2n}
H_{2,\de}(\mu;v)\cK_2(v,w,z)=H_{2,\de}(\mu;w)\cK_2(v,w,z)=H_{2,\de}(\mu;z)\cK_2(v,w,z),\ \ \ \de=+,-,
\ee
where
\be\label{cK2}
\cK_2(v,w,z):=(W(v)W(w)W(z))^{1/2}\cS_2(v,w,z).
\ee
Here we are thinking of vectors $v,w,z$ restricted by~\eqref{xord}, and once more taking positive square roots. 

Requiring the ordering~\eqref{xord} and the $A_2$ sum constraint,  we obtain a set denoted $G_2$. When we fix $z\in G_2$, \eqref{cK2} yields a function that can be reinterpreted as the kernel of a Hilbert-Schmidt integral operator on
$L^2(G_2,dw_r)$, where $dw_r$ denotes Lebesgue measure on $G_2$. (The subscript~$r$ encodes the hyperplane restriction at issue.) Specifically, we are thinking of the integral operator defined by
\be\label{Iz}
(I(z)f)(v):= \int_{G_2}\cK_2(v,-w,z)f(w)dw_r,\ \ \ z\in G_2.
\ee
 We expect that this yields a commutative family of Hilbert-Schmidt operators, all of whose singular values are (generically) positive. The associated orthonormal functions can be expected to yield a basis of positive-eigenvalue eigenfunctions for the Hamiltonians, too, allowing them to be reinterpreted as bona fide self-adjoint operators. (More information on this scenario can be found in previous work that can be traced from~\cite{Ru15}.) 

We proceed to render these expectations more precise in the following conjecture.

\begin{conjecture}
There exists a convergent expansion
\be\label{expans}
\cS_2(v,w,z)=\sum_{n=0}^{\infty} c_n J_n(v)J_n(w)J_n(z),\ \ c_n>0,\ \ \ v,w,z\in G_2,
\ee
with $J_n(x)$, $n\in\N$, functions on $G_2$  with a meromorphic extension to 
\be\label{cM2}
\cM_2:=\{ x\in\C^3 \mid x_1+x_2+x_3=0\},
\ee
 and with the following additional features:
 \be\label{Jncon}
 J_n(-x)=\overline{J_n(\overline{x})},\ \ \ x\in\cM_2,
 \ee
\be
\sum_{n=0}^{\infty}|c_nJ_n(x)|^2 <\infty,\ \ \ x\in G_2,
\ee
\be
A_{2,\de}(\mu;x)J_n(x)=\lambda_{\de,n}(\mu)J_n(x),\ \ \de=+,-,\ \ \ n\in\N,\ \ x\in\cM_2,
\ee 
\be
\lambda_{\de,n}(\mu)>0,\ \ \ \mu\in \R \cup i\R .
\ee
Moreover, the set of vectors $J_n(x)$, $n\in\N$, yields an orthonormal base for $L^2(G_2,W(x)dx_r)$. 
\end{conjecture}

A concrete 1-variable counterpart for the state of affairs encoded in this conjecture is provided by Theorem~2.2 in~\cite{Ru13}. In that case, however, the orthonormal functions are real-valued. As already suggested by~\eqref{Jncon}, we expect that this not the case for the functions $J_n(x)$, $x\in G_2$. Indeed, the Hilbert-Schmidt integral operator $I(z)$~\eqref{Iz} does not have a real-valued kernel and is not self-adjoint.

On the other hand, we stress that the Hilbert-Schmidt property suffices for the existence of a singular value decomposition of the kernel. Combining this with the obvious $S_3$ symmetry of the kernel, it may well be possible to prove that this already entails the expansion~\eqref{expans}, with orthonormal functions $J_n(x)$ and with $c_n\ge 0$.  
However, to show that the operator is not only infinite-rank, but complete in the sense of~\cite{Ru15} (i.~e., that all $c_n$ are nonzero and the functions yield an orthonormal \emph{basis}), requires arguments beyond those provided by~\cite{Ru13}. 

As a final remark on this elliptic conjecture, we point out that it may turn out to be more natural to label the basis functions not by $n\in\N$, but rather by a multi-index
\be\label{kcon}
k\in\Z^3,\ \ \ k_1+k_2+k_3=0.
\ee

For the hyperbolic case the kernel identities are once more given by~\eqref{kid2}, now with $G$ in~\eqref{kf2} the hyperbolic gamma function (cf.~Appendix~A), and the A$\De$O family given by
\begin{multline}\label{A2hsum} 
A_{2,\de}(\mu;x)\equiv \frac{c_{\de}(x_2-x_3\pm\mu)}{s_{\de}(x_1-x_2)s_{\de}(x_1-x_3)}\exp\Big(\frac13 ia_{-\de}\big(2\partial_{x_1}-\partial_{x_2}-\partial_{x_3}\big)\Big)
\\
+  \frac{c_{\de}(x_3-x_1\pm\mu)}{s_{\de}(x_2-x_3)s_{\de}(x_2-x_1)}\exp\Big(\frac13 ia_{-\de}\big(2\partial_{x_2}-\partial_{x_3}-\partial_{x_1}\big)\Big)
\\
+ \frac{c_{\de}(x_1-x_2\pm\mu)}{s_{\de}(x_3-x_1)s_{\de}(x_3-x_2)}\exp\Big(\frac13 ia_{-\de}\big(2\partial_{x_3}-\partial_{x_1}-\partial_{x_2}\big)\Big).
\end{multline}
(Here we use the hyperbolic functions~\eqref{hnot}.) In this case the kernel identity involves the equality of the functions
\be\label{cLhn} 
\cL\equiv \frac{c_+(v_2-v_3\pm\mu)}{s_+(v_1-v_2)s_+(v_1-v_3)}
 \prod_{k\ne 1,l,m}\frac{1}{s_+(v_k+w_l+z_m+ia_+/3)}+\mathrm{cyclic},
\ee
and
\be\label{cRhn}
\cR\equiv \frac{c_+(w_2-w_3\pm\mu)}{s_+(w_1-w_2)s_+(w_1-w_3)}
 \prod_{k,l\ne 1,m}\frac{1}{s_+(v_k+w_l+z_m+ia_+/3)}+\mathrm{cyclic}.
\ee
More precisely, just as in the elliptic case, these functions are not equal for unconstrained $v,w,z\in\C^3$, but they become equal under the $A_2$ sum constraint~\eqref{A2con}. 

With one exception, the assertions just made readily follow from their elliptic counterparts by using the limit relations between the elliptic and hyperbolic quantities that can be found in~Subsection~III~B of~\cite{Ru97}. This exception is the assertion $\cL\ne \cR$ for unconstrained variables. (Note that the residue argument encoded in~\eqref{sumne1}--\eqref{sumne2} is inconclusive, since $r=0$ in the hyperbolic case.) We proceed to supply a proof.

\begin{proposition}
Letting $v,w,z\in\C^3$, the functions $\cL$ and $\cR$ are not equal.
\end{proposition}
\begin{proof}
We multiply $\cL$ and $\cR$ by
\be
\prod_{j=1}^2s_+(v_j+w_j+z_j+ia_+/3),
\ee
and then set
\be
v_j+w_j+z_j+ia_+/3=0,\ \ \ j=1,2.
\ee
For $\cL$ this yields
\be\label{Lm}
\frac{c_+(v_1-v_2\pm\mu)}{s_+(v_3-v_1)s_+(v_3-v_2)}\prod_{k=1,2}\prod_{\stackrel{l,m=1,2,3}{(l,m)\ne(k,k)}}\frac{1}{s_+(v_k+w_l+z_m+ia_+/3)},
\ee
whereas for $\cR$ we get
\be\label{Rm}
\frac{c_+(w_1-w_2\pm\mu)}{s_+(w_3-w_1)s_+(w_3-w_2)}\prod_{l=1,2}\prod_{\stackrel{k,m=1,2,3}{(k,m)\ne(l,l)}}\frac{1}{s_+(v_k+w_l+z_m+ia_+/3)}.
\ee
Next, multiplying~\eqref{Lm}--\eqref{Rm} by $s_+(v_3+w_1+z_2+ia_+/3)$ and then setting $v_3+w_1+z_2+ia_+/3=0$, the first product vanishes, whereas the second one is nonzero. Thus the assertion follows..
\end{proof}

Both with and without the $A_2$ sum constraint, the operators~\eqref{A2hsum}  are a commuting family (i.~e., they satisfy~\eqref{Acom}). Indeed, this is easily checked by taking the limit $r\to 0$ in the elliptic objects.
Also, assuming from now on
\be
x_3<x_2<x_1,
\ee
the weight function $W(x)$~\eqref{W} (with $G$ the hyperbolic gamma function) becomes the positive function
\be\label{Wh}
W(x)=\prod_{1\le j<k\le 3}4s_+(x_j-x_k)s_-(x_j-x_k),
\ee
and letting $\mu\in\R$ or $\mu\in i\R $,  the A$\De$Os are formally positive with respect to the measure $W(x)dx$. Indeed, the similarity transformation~\eqref{Hde} yields the operators
 \begin{multline}\label{H2hsum} 
H_{2,\de}(\mu;x)= c_{\de}(x_2-x_3\pm\mu)V_{1,\de}(x)^{1/2}\exp\Big(\frac13 ia_{-\de}\big(2\partial_{x_1}-\partial_{x_2}-\partial_{x_3}\big)\Big)V_{1,\de}(x)^{1/2}
\\
+ c_{\de}(x_3-x_1\pm\mu)V_{2,\de}(x)^{1/2}\exp\Big(\frac13 ia_{-\de}\big(2\partial_{x_2}-\partial_{x_3}-\partial_{x_1}\big)\Big)V_{2,\de}(x)^{1/2}
\\
+ c_{\de}(x_1-x_2\pm\mu)V_{3,\de}(x)^{1/2}\exp\Big(\frac13 ia_{-\de}\big(2\partial_{x_3}-\partial_{x_1}-\partial_{x_2}\big)\Big) V_{3,\de}(x)^{1/2} ,
\end{multline}   
where the coefficients $V_{j,\de}(x)$ are again defined by~\eqref{V1}--\eqref{V3}, now with $s_{\de}(x)$ given by~\eqref{hnot}. This explicit representation can be checked in the same way as in the elliptic case, with \eqref{sRy} and~\eqref{ssy} replaced by
\be\label{sRyh}
\frac{\sqrt{s_{\de}(y)}}{c_{\de}(sy-ia_{\de}/2)}= \frac{si}{\sqrt{s_{\de}(y)}},\ \ \ \ y>0, \ \  \ \ s=+,-,
\ee
and
\be\label{ssyh}
\exp(sia_{-\de}d/dy)\frac{1}{\sqrt{s_{-\de}(y)}}=-si \frac{1}{\sqrt{s_{-\de}(y)}}\exp(sia_{-\de}d/dy), \ \ y>0, \ \ \ \ \ s=+,-.
\ee  
    
Defining next
\be
G_2:=\{ x\in\R^3 \mid x_1+x_2+x_3=0,\ \ x_3<x_2<x_1\},
\ee
we conjecture that the hyperbolic version of $I(z)$~\eqref{Iz}  does not define a Hilbert-Schmidt integral operator on $L^2(G_2,dw_r)$. (Equivalently, the hyperbolic version $\cK_2(v,w,z)$ of the kernel~\eqref{cK2} with $z$ fixed is most likely not square-integrable over~$G_2\times G_2$.) However, $I(z)$ may well be a bounded integral operator.

In any event, a natural counterpart of the above elliptic conjecture is the following hyperbolic conjecture.

\begin{conjecture}
Letting
\be
 \hat{G}_2\equiv \{k\in\R^3\mid k_1+k_2+k_3=0\},
 \ee
 there exists a convergent expansion
\be\label{expansh}
\cS_2(v,w,z)=\int_{\hat{G}_2}c_k J_k(v)J_k(w)J_k(z)dk_r,\ \ c_k>0,\ \ \  v,w,z\in G_2, 
\ee
where $J_k(x)$ (with $k\in \hat{G}_2$ fixed) is a function on~$G_2$ with a meromorphic continuation to~$\cM_2$~\eqref{cM2} 
 and with the following additional features:
  \be\label{Jkcon}
 J_k(-x)=\overline{J_k(\overline{x})},\ \ \ x\in\cM_2,
 \ee
 \be
A_{2,\de}(\mu;x)J_k(x)=\lambda_{\de,k}(\mu)J_k(x),\ \ \de=+,-,\ \ x\in\cM_2,
\ee
\be
\lambda_{\de,k}(\mu)>0,\ \ \  \ \mu\in \R \cup i\R .
\ee
Moreover, the functions $J_k(x)$ yield the kernel of an isometric integral transformation from $L^2(G_2,W(x)dx_r)$ onto $L^2(\hat{G}_2,dk_r)$. 
\end{conjecture}

An explicit 1-variable example illustrating this hyperbolic conjecture can be found in Section~2 of~\cite{HaRu18}, see in particular Eq.~(92) on p.~213.  

We proceed to point out a remarkable connection to an A$\De$O associated with the dual relativistic nonperiodic (open)  Toda 3-particle system. To this end we write
\be
2A_{2,\de}(\mu;x)=c_{\de}(2\mu)B_{\de}(x)+C_{\de}(x),
\ee
where
\be
D_{\de}(x):=\sum_{m=1}^3 D_{\de}^{(m)}(x),\ \ \ D=B,C,
\ee    
with
\be
B_{\de}^{(1)}(x):= \frac{ 1}{s_{\de}(x_1-x_2)s_{\de}(x_1-x_3)}\exp\Big(\frac13 ia_{-\de}\big(2\partial_{x_1}-\partial_{x_2}-\partial_{x_3}\big)\Big),
\ee
\be
C_{\de}^{(1)}(x):= \frac{c_{\de}(2x_2-2x_3)}{s_{\de}(x_1-x_2)s_{\de}(x_1-x_3)}\exp\Big(\frac13 ia_{-\de}\big(2\partial_{x_1}-\partial_{x_2}-\partial_{x_3}\big)\Big),
\ee
and $D_{\de}^{(m)}(x)$, $m=2,3$, obtained from $D_{\de}^{(1)}(x)$ by cyclic permutations. The relation now reads
\be\label{BATr}
B_{\de}(x)=-4\hat{A}_{-1,\de}(x)\exp\Big(\frac13 ia_{-\de}\big(-\partial_{x_1}-\partial_{x_2}-\partial_{x_3}\big)\Big),
\ee
where the A$\De$O  $\hat{A}_{-1,\de}(x)$ is given by Eq.~(2.132) in~\cite{HaRu12} (with $N=3$). The latter satisfies two kernel identities involving the dual relativistic Toda kernel functions
\be
\cT_{\tau}^{\mathrm{rel}}(x,y):=\prod_{j,k=1}^3 G(\tau(x_j-y_k)),\ \ \ \ \tau=+,-,
\ee
where $x,y\in\C^3$, cf.~Theorem~2.10 in~\cite{HaRu12}. As a corollary of this theorem, it follows that $B_{\de}(x)$ satisfies two kernel identities
\be\label{Bkid}
B_{\de}(x)\cT_{\tau}^{\mathrm{rel}}(x,-y)=B_{\de}(y)\cT_{\tau}^{\mathrm{rel}}(x,-y),\ \ \ \tau=+,-,
\ee
for unconstrained variables, in addition to the kernel identities
\be
B_{\de}(x)\cS_{2}(x,y,z)=B_{\de}(y)\cS_{2}(x,y,z)=B_{\de}(z)\cS_{2}(x,y,z),
\ee
which only hold for $x,y,z$ obeying the $A_2$ sum constraint. We also note that the weight function $W(x)$~\eqref{Wh} is equal to the dual relativistic Toda weight function given by Eq.~(2.148) in~\cite{HaRu12}. However, it is not hard to see that $C_{\de}(x)$ does not satisfy the Toda kernel identities~\eqref{Bkid}.

From this state of affairs it can be expected that the Hilbert space reinterpretation of the A$\De$O $B_{\de}(x)$ associated with the above open conjecture will be quite different from the one arising from the dual relativistic nonperiodic  Toda 3-particle system. (To be sure, the problem of associating commuting self-adjoint Hilbert space operators to the commuting Toda A$\De$Os $\hat{A}_{-1,\pm}(x)$  is to date still open, too.)

We continue by briefly studying the trigonometric regime, which arises upon taking
\be\label{subt}
a_+\to -i\pi/r,\ \ \ a_-\to\alpha.
\ee
For the hyperbolic functions $\cL$ and $\cR$ given by~\eqref{cLh}--\eqref{cRh}, this yields trigonometric functions
\be\label{cLt} 
\lambda\equiv \frac{\cos r(v_2-v_3\pm\mu)}{\sin r(v_1-v_2)\sin r(v_1-v_3)}
 \prod_{k\ne 1,l,m}\frac{1}{\sin r(v_k+w_l+z_m+\pi/3r)}+\mathrm{cyclic},
\ee
and
\be\label{cRt}
\rho\equiv \frac{\cos r(w_2-w_3\pm\mu)}{\sin r(w_1-w_2)s_+(w_1-w_3)}
 \prod_{k,l\ne 1,m}\frac{1}{\sin r(v_k+w_l+z_m+\pi/3r)}+\mathrm{cyclic}.
\ee
As follows from our previous results, these functions are equal in case $v,w,z$ satisfy the $A_2$ sum constraint. Can this equality be reinterpreted as a trigonometric kernel identity? 

To try and answer this, we note that  the substitutions~\eqref{subt} in~\eqref{A2hsum} with $\de=+$ yield the A$\De$O
\be\label{Atraux}
 \frac{\cos r(x_2-x_3\pm\mu)}{\sin r(x_1-x_2)\sin r(x_1-x_3)}\exp\Big(\frac13 i\alpha\big(2\partial_{x_1}-\partial_{x_2}-\partial_{x_3}\big)\Big)+\mathrm{cyclic}.
\ee
Also, setting 
\be
\cS_t\equiv \prod_{k,l,m=1}^3G_t(v_k+w_l+z_m-\de_t),\ \ \ \de_t\equiv i\alpha/6-\pi/3r,
\ee
we obtain
\begin{multline}
\frac{\exp\Big(\frac13 i\alpha\big(2\partial_{v_1}-\partial_{v_2}-\partial_{v_3}\big)\Big)\cS_t}{\prod_{k,l,m=1}^3G_t(v_k-i\alpha/3+w_l+z_m-\de_t)}=\prod_{l,m=1}^3\frac{G_t(v_1+2i\alpha/3+w_l+z_m-\de_t)}{G_t(v_1-i\alpha/3+w_l+z_m-\de_t)}
\\
=\prod_{l,m=1}^3[1-\exp(2ir(v_1+w_l+z_m+\pi/3r))],
\end{multline}
where we used the $G_t$-A$\De$E~\eqref{Gtade}. Thanks to the sum constraint, this can be rewritten as
\be
 C\exp(9irv_1)\prod_{l,m} \sin r(v_1+w_l+z_m+\pi/3r),
\ee
with $C$ an irrelevant constant. From this we deduce that if the factor $\exp(9irv_1)$ were absent, we would arrive at kernel identities of the previous type for the A$\De$O~\eqref{Atraux} and kernel function~$\cS_t$.

The only way we see to take this extra factor into account is to introduce the trigonometric A$\De$O
\be\label{Atr}
A_t(\mu;x)\equiv  \frac{\exp(-9irx_1)\cos r(x_2-x_3\pm\mu)}{\sin r(x_1-x_2)\sin r(x_1-x_3)}\exp\Big(\frac13 i\alpha\big(2\partial_{x_1}-\partial_{x_2}-\partial_{x_3}\big)\Big)+\mathrm{cyclic}.
\ee
Indeed, for this choice we do obtain the kernel identities
\be
A_t(\mu;v)\cS_t(v,w,z)=A_t(\mu;w)\cS_t(v,w,z)=A_t(\mu;z)\cS_t(v,w,z), \ \ \ \sum_{j=1}^3 x_j=0,\ \ x=v,w,z.
\ee
However, we are not aware of a weight function $W_t(x)$ so that (with $\mu$ real or purely imaginary) the A$\De$O $A_t(\mu;x)$ becomes formally s.\,a.~with respect to the measure $W_t(x)dx$.

There seems to be no non-trivial rational version of the kernel identities. To be sure, there is a natural rational counterpart of the hyperbolic A$\De$O $A_{2,+}(\mu;x)$ (say) and the trigonometric A$\De$O $A_t(\mu;x)$. Indeed, letting $\alpha:=a_-$ and taking $a_+\to\infty$ in the first case or $r\to 0$ in the second one, obvious renormalizations yield the limit
\begin{multline}\label{A2rsum} 
A_r(x)\equiv \frac{1}{(x_1-x_2)(x_1-x_3)}\exp\Big(\frac13 i\alpha \big(2\partial_{x_1}-\partial_{x_2}-\partial_{x_3}\big)\Big)
\\
+  \frac{1}{(x_2-x_3)(x_2-x_1)}\exp\Big(\frac13 i\alpha \big(2\partial_{x_2}-\partial_{x_3}-\partial_{x_1}\big)\Big)
\\
+ \frac{1}{(x_3-x_1)(x_3-x_2)}\exp\Big(\frac13 i\alpha \big(2\partial_{x_3}-\partial_{x_1}-\partial_{x_2}\big)\Big).
\end{multline}

At face value, the limit of the equality of the (restricted) hyperbolic functions $\cL$ and $\cR$ (cf.~\eqref{cL}--\eqref{cR}) and their trigonometric counterparts $\lambda$ and $\rho$ (cf.~\eqref{cLt}--\eqref{cRt}) seems strange, as it says that a function of $v$ becomes equal to a function of $w$. The explanation is that both functions vanish. Indeed, we have an identity
\be\label{ratid}
 \frac{1}{(x_1-x_2)(x_1-x_3)}+ \frac{1}{(x_2-x_3)(x_2-x_1)} + \frac{1}{(x_3-x_1)(x_3-x_2)}=0,
\ee
which clearly holds for all $x\in\C^3$. Put differently, the constant function is an eigenfunction of $A_r(x)$ with eigenvalue zero.

In fact, the rational A$\De$O is related to an A$\De$O associated with the 
dual nonrelativistic nonperiodic  Toda 3-particle system. Specifically, we have
\be\label{AATn}
A_r(x)=-\alpha^{-2}\hat{A}_{-1,\mathrm{nr}}(x)\exp\Big(\frac13 i\alpha \big(-\partial_{x_1}-\partial_{x_2}-\partial_{x_3}\big)\Big),
\ee
where the A$\De$O  $\hat{A}_{-1,\mathrm{nr}}(x)$ is given by Eq.~(4.112) in~\cite{HaRu12} (with $N=3$ and $\hbar\mu=\alpha$). From Theorem~4.11 in~\cite{HaRu12} it then readily follows that $A_r(x)$ satisfies two kernel identities
\be
A_r(x)\cT_{\sigma}^{\mathrm{nr}}(x,-y)=A_r(y)\cT_{\sigma}^{\mathrm{nr}}
(x,-y),\ \ \ \sigma=1,-1,
\ee
for unconstrained variables, with dual nonrelativistic kernel functions given by
\be
\cT_{\sigma}^{\mathrm{nr}}(x,y):=\prod_{j,k=1}^3 \Gamma((x_j-y_k)/\alpha)^{\sigma},\ \ \ \ \sigma=\pm 1.
\ee


\subsection{The $A_3$ case}

Proceeding along the same lines as in Subsection~4.1, we first note that we have
\be
[A_{3,+}(x),A_{3,-}(x)]=0,
\ee
where
 the A$\De$Os $A_{3,\pm}(x)$ are given by \eqref{A3sum}--\eqref{A31}. (Using \eqref{Rade} this commutation relation is easily verified.) 
 
 Next, letting $x\in\R^4$, it is clear from~\eqref{defR} that the numerator  factors $R_{\de}(x_j-x_k)$ in the coefficients take values in $(0,\infty)$. 
The shifts are formally   s.\,a.~(self-adjoint)  with respect to Lebesgue measure on $\R^4$. They commute with the numerator functions, but they do not commute with the (complex-valued)  denominator functions. Hence the A$\De$Os are not  formally s.\,a. with respect to Lebesgue measure. 

This can be remedied in the same way as in the $A_2$ case (cf.~also~\cite{Ra18}): The counterparts of \eqref{W}--\eqref{Hde} are the formulas
 \be\label{W3}
W(x):=\prod_{1\le j<k\le 4}G(\pm (x_j-x_k)+ia),
\ee
 \be
W(x)=(p_+p_-)^6\prod_{1\le j<k\le 4}s_+(x_j-x_k)s_-(x_j-x_k),
\ee
\be\label{xord3}
-\pi/2r\le x_4<x_3<x_2<x_1<\pi/2r,
\ee
 \be\label{Hde3}
H_{3,\de}(x)    \equiv  W(x)^{1/2}A_{3,\de}(x)W(x)^{-1/2} ,\ \ \ \de=+,-,
\ee
and then we have the following result.

\begin{proposition}
With the ordering~\eqref{xord3} in force, the A$\De$Os $H_{3,\de}(x)$ are formally positive with respect to Lebesgue measure.
\end{proposition}
\begin{proof} 
Using the relations~\eqref{sRy} and~\eqref{ssy}, we can explicitly determine the A$\De$Os, yielding
\begin{multline}\label{H3de}
H_{3,\de}(x)  =   \frac{\exp(-3ra_{-\de}/2)}{p_{\de}^3}\Big[ R_{\de}(x_2-x_3)R_{\de}(x_3-x_4)R_{\de}(x_4-x_2)
\\
 \times V_{1,\de}(x)^{1/2}\exp\Big(\frac14 ia_{-\de}\big(3\partial_{x_1}-\partial_{x_2}-\partial_{x_3}-\partial_{x_4}\big)\Big)V_{1,\de}(x)^{1/2}
  \\
 +R_{\de}(x_3-x_4)R_{\de}(x_4-x_1)R_{\de}(x_1-x_3)
 V_{2,\de}(x)^{1/2}\exp\Big(\frac14 ia_{-\de}\big(3\partial_{x_2}-\partial_{x_3}-\partial_{x_4}-\partial_{x_1}\big)\Big)V_{2,\de}(x)^{1/2}
 \\
 +R_{\de}(x_4-x_1)R_{\de}(x_1-x_2)R_{\de}(x_2-x_4)
 V_{3,\de}(x)^{1/2}\exp\Big(\frac14 ia_{-\de}\big(3\partial_{x_3}-\partial_{x_4}-\partial_{x_1}-\partial_{x_2}\big)\Big)V_{3,\de}(x)^{1/2}
 \\
 +R_{\de}(x_1-x_2)R_{\de}(x_2-x_3)R_{\de}(x_3-x_1)
 V_{4,\de}(x)^{1/2}\exp\Big(\frac14 ia_{-\de}\big(3\partial_{x_4}-\partial_{x_1}-\partial_{x_2}-\partial_{x_3}\big)\Big)V_{4,\de}(x)^{1/2}\Big], 
\end{multline} 
where
\be\label{V31}
V_{1,\de}(x):=1/s_{\de}(x_1-x_2)s_{\de}(x_1-x_3)s_{\de}(x_1-x_4),
\ee
\be\label{V32}
V_{2,\de}(x):=1/s_{\de}(x_1-x_2)s_{\de}(x_2-x_3)s_{\de}(x_2-x_4),
\ee
\be\label{V33}
V_{3,\de}(x):=1/s_{\de}(x_1-x_3)s_{\de}(x_2-x_3)s_{\de}(x_3-x_4),
\ee
\be\label{V34}
V_{4,\de}(x):=1/s_{\de}(x_1-x_4)s_{\de}(x_2-x_4)s_{\de}(x_3-x_4),
\ee
and positive square roots are understood. From this positivity can be read off. 
\end{proof}

Imposing the $A_3$ sum constraint from now on, we expect that the Hamiltonians \eqref{Hde3}
satisfy the kernel identities
\be\label{kid3n}
H_{3,\de}(v)\cK_3(d;v,w)=H_{3,\de}(w)\cK_3(d;v,w),\ \ \ \de=+,-,
\ee
with
\be\label{cK3}
\cK_3(d;v,w):=(W(v)W(w))^{1/2}\cS_3(d;v,w).
\ee
(Here the vectors $v,w$ are assumed to satisfy~\eqref{xord3}, and positive square roots are understood.) Indeed, the validity of these identities is equivalent to the (expected) validity of the elliptic kernel identities~\eqref{kid3}.

With the ordering~\eqref{xord3} and the $A_3$ sum constraint in effect,  we get a set $G_3$. Fixing~$d$ in the strip $\{|\im d|<a/2\}$,  \eqref{cK3} yields a function that can be viewed as the kernel of a Hilbert-Schmidt integral operator on
$L^2(G_3,dw_r)$ defined by (recall $a:=(a_++a_-)/2$)
\be\label{Id}
(I(d)f)(v):= \int_{G_3}\cK_3(d;v,-w)f(w)dw_r,\ \ \ |\im d|<a/2.
\ee
We are now prepared for the following conjecture.

\begin{conjecture}
Letting $|\im d|<a/2$, there exists a convergent expansion
\be\label{expans3}
\cS_3(d;v,w)=\sum_{n=0}^{\infty} c_n(d) J_n(v)J_n(w) ,\ \ c_n(d)\in \C^*,\ \ \ v,w\in G_3,
\ee
with $J_n(x)$, $n\in\N$, functions on $G_3$  with a meromorphic extension to 
\be\label{cM3}
\cM_3:=\{ x\in\C^4 \mid x_1+x_2+x_3+x_4=0\},
\ee
 and with the following additional features:
 \be\label{Jncon3}
 J_n(-x)=\overline{J_n(\overline{x})},\ \ \ x\in\cM_3,
 \ee
\be
\sum_{n=0}^{\infty}|c_n(d)|^2 <\infty,
\ee
\be
A_{3,\de}(x)J_n(x)=\lambda_{\de,n}J_n(x),\ \ \de=+,-,\ \ \ n\in\N,\ \ x\in\cM_3,
\ee 
\be
\lambda_{\de,n}>0.
\ee
Moreover, the set of vectors $J_n(x)$, $n\in\N$, yields an orthonormal base for $L^2(G_3,W(x)dx_r)$. 
\end{conjecture}

Observe that it would follow from this conjecture that the Hilbert-Schmidt operators $I(d)$~\eqref{Id} form a commutative family. We also note that it may be more natural to trade $n\in\N$ for a multi-index
\be
k\in\Z^4,\ \ \ k_1+k_2+k_3+k_4=0.
\ee

For the hyperbolic case the weight function~\eqref{W3} becomes
 \be\label{W3h}
W(x)=\prod_{1\le j<k\le 4}4s_+(x_j-x_k)s_-(x_j-x_k),
\ee
the $s_{\de}$ now being given by~\eqref{hnot}. The hyperbolic counterpart of~\eqref{H3de} (again denoted $H_{3,\de}(x)$) equals
\begin{multline}\label{H31} 
c_{\de}(x_2-x_3)c_{\de}(x_3-x_4)c_{\de}(x_4-x_2)
 V_{1,\de}(x)^{1/2}\exp\Big(\frac14 ia_{-\de}\big(3\partial_{x_1}-\partial_{x_2}-\partial_{x_3}-\partial_{x_4}\big)\Big)V_{1,\de}(x)^{1/2}
  \\
 +c_{\de}(x_3-x_4)c_{\de}(x_4-x_1)c_{\de}(x_1-x_3)
 V_{2,\de}(x)^{1/2}\exp\Big(\frac14 ia_{-\de}\big(3\partial_{x_2}-\partial_{x_3}-\partial_{x_4}-\partial_{x_1}\big)\Big)V_{2,\de}(x)^{1/2}
 \\
 +c_{\de}(x_4-x_1)c_{\de}(x_1-x_2)c_{\de}(x_2-x_4)
 V_{3,\de}(x)^{1/2}\exp\Big(\frac14 ia_{-\de}\big(3\partial_{x_3}-\partial_{x_4}-\partial_{x_1}-\partial_{x_2}\big)\Big)V_{3,\de}(x)^{1/2}
 \\
 +c_{\de}(x_1-x_2)c_{\de}(x_2-x_3)c_{\de}(x_3-x_1)
 V_{4,\de}(x)^{1/2}\exp\Big(\frac14 ia_{-\de}\big(3\partial_{x_4}-\partial_{x_1}-\partial_{x_2}-\partial_{x_3}\big)\Big)V_{4,\de}(x)^{1/2}, 
\end{multline} 
with $V_{m,\de}(x)$, $m=1,2,3,4$, defined by~\eqref{V31}--\eqref{V34} and~\eqref{hnot}. Clearly, these Hamiltonians satisfy the kernel identities~\eqref{kid3n}--\eqref{cK3}.

Introducing
\be
G_3:=\{ x\in\R^4 \mid x_1+x_2+x_3+x_4=0,\ \ x_4<x_3<x_2<x_1\},
\ee
we are prepared for a hyperbolic analog of the above elliptic conjecture.

\begin{conjecture}
Letting
\be
 \hat{G}_3\equiv \{k\in\R^4\mid k_1+k_2+k_3+k_4=0\},
 \ee
 there exists a convergent expansion
\be\label{expansh3}
\cS_3(d;v,w)=\int_{\hat{G}_3}c_k(d) J_k(v)J_k(w)dk_r,\ \ c_k(d)\in\C^*,\ \ \  v,w\in G_3, \ \ |\im d|<a/2,
\ee
where $J_k(x)$ (with $k\in \hat{G}_3$ fixed) is a function on~$G_3$ with a meromorphic continuation to~$\cM_3$~\eqref{cM3} 
 and with the following additional features:
  \be\label{Jkcon3}
 J_k(-x)=\overline{J_k(\overline{x})},\ \ \ x\in\cM_3,
 \ee
 \be
A_{3,\de}(x)J_k(x)=\lambda_{\de,k}J_k(x),\ \ \de=+,-,\ \ x\in\cM_3,
\ee
\be
\lambda_{\de,k}>0.
\ee
Moreover, the functions $J_k(x)$ yield the kernel of an isometric integral transformation from $L^2(G_3,W(x)dx_r)$ onto $L^2(\hat{G}_3,dk_r)$. 
\end{conjecture}

Recalling the identity~\eqref{iden0} and the definition of the hyperbolic A$\De$Os $A_{3,\pm}(x)$ (encoded in~\eqref{A3sum} and~\eqref{A31h}), it is immediate that the constant function is a zero-eigenvalue eigenfunction. Clearly it is not in $L^2(G_3,W(x)dx_r)$, but that is not going to be true for the functions $J_k(x)$ either. However, it might still have a special relation to $\cS_3(d;v,w)$ and $W(x)$. Possibly, the integral 
\be
\int_{G_3} \cS_3(d;v,w)W(w)dw_r
\ee
converges in a suitable sense and does not depend on~$v$.

Turning to the trigonometric regime, the substitutions~\eqref{subt} in the functions~$\cL$~\eqref{cLh} and $\cR$~\eqref{cRh} yield functions
\be
\lambda \equiv \frac{\cos r(v_2-v_3)\cos r(v_3-v_4)\cos r(v_4-v_2)}{\prod_{j>1}\sin r(v_1-v_j) }
 \prod_{k\ne 1,l}\frac{1}{\sin r(v_k+w_l+\pi/4r\pm d)}+\mathrm{cyclic}.
\ee
and
\be
\rho\equiv \frac{\cos r(w_2-w_3)\cos r(w_3-w_4)\cos r(w_4-w_2)}{\prod_{j>1}\sin r(w_1-w_j)  } \prod_{k,l\ne 1}\frac{1}{\sin r(v_k+w_l+\pi/4r\pm d )}+\mathrm{cyclic}.
\ee
These functions become equal when the $A_3$ sum constraint is imposed. Again, we can ask whether this equality can be reinterpreted as a kernel identity.

To study this, we first note that the obvious trigonometric counterpart of the A$\De$O $A_{3,+}(x)$ is given by
\be\label{Atraux3}
\frac{\cos r(x_2-x_3)\cos r(x_3-x_4)\cos r(x_4-x_2)}{\prod_{j>1}\sin r(x_1-x_j) }
\exp\Big(\frac14 i\alpha \big(3\partial_{x_1}-\partial_{x_2}-\partial_{x_3}-\partial_{x_4}\big)\Big)+\mathrm{cyclic}.
\ee
Letting
\be
\cS_t\equiv \prod_{k,l=1}^4G_t(v_k+w_l-\de_t\pm d),\ \ \ \de_t\equiv i\alpha/4-\pi/4r,
\ee
we obtain using \eqref{Gtade}
\begin{multline}
\frac{\exp\Big(\frac14 i\alpha\big(3\partial_{v_1}-\partial_{v_2}-\partial_{v_3}-\partial_{v_4}\big)\Big)\cS_t}{\prod_{k,l=1}^4G_t(v_k-i\alpha/4+w_l-\de_t\pm d)}=\prod_{l=1}^4\frac{G_t(v_1+3i\alpha/4+w_l-\de_t\pm d)}{G_t(v_1-i\alpha/4+w_l-\de_t\pm d)}
\\
=\prod_{l=1}^4[1-\exp(2ir(v_1+w_l+\pi/4r\pm d))].
\end{multline}
 Using the $A_3$ sum constraint, this becomes
\be
 C\exp(8irv_1)\prod_{l=1}^4 \sin r(v_1+w_l+\pi/4r\pm d),
\ee
with $C$ a constant. Therefore, without the factor $\exp(8irv_1)$, we would  arrive at kernel identities of the previous type for the A$\De$O~\eqref{Atraux3} and kernel function~$\cS_t$.

We can take the additional factor into account via a trigonometric A$\De$O $A_t(x)$ defined by
\be\label{At3} 
\frac{e^{-8irx_1}\cos r(x_2-x_3)\cos r(x_3-x_4)\cos r(x_4-x_2)}{\prod_{j>1}\sin r(x_1-x_j) }
\exp\Big(\frac14 i\alpha \big(3\partial_{x_1}-\partial_{x_2}-\partial_{x_3}-\partial_{x_4}\big)\Big)+\mathrm{cyclic}.
\ee
Then we obtain kernel identities
\be\label{Atid}
A_t(v)\cS_t(d;v,w)=A_t(w)\cS_t(d;v,w),\ \ \ \ \sum_{j=1}^4 x_j=0,\ \ x=v,w.
\ee
However, just as in the $A_2$ case, we do not know a weight function $W_t(x)$ such that $A_t(x)$ can be viewed as a formally s.\,a.~operator with respect to the measure~$W_t(x)dx$.

We are also not aware of a non-trivial rational version of the $A_3$ kernel identities. The natural rational counterpart of the hyperbolic A$\De$O $A_{3,+}(x)$ and the trigonometric A$\De$O $A_t(x)$ is obtained by taking $\alpha:=a_-$ and then letting $a_+\to\infty$ in the first case or $r\to 0$ in the second one. With obvious renormalizations this yields 
\be\label{A3rsum} 
A_r(x)\equiv \frac{1}{(x_1-x_2)(x_1-x_3)(x_1-x_4)}\exp\Big(\frac14 i\alpha \big(3\partial_{x_1}-\partial_{x_2}-\partial_{x_3}-\partial_{x_4}\big)\Big)
+\mathrm{cyclic}.
\ee 
As before, the limit of the equality of the (restricted) hyperbolic functions $\cL$ and $\cR$  and their trigonometric counterparts $\lambda$ and $\rho$ can be understood from the counterpart of the identity~\eqref{ratid}, viz.,
\be
 \frac{1}{(x_1-x_2)(x_1-x_3)(x_1-x_4)}+\mathrm{cyclic}=0,\ \ x\in\C^4,
 \ee
whose validity is easily verified directly. Hence, the constant function is again an eigenfunction of $A_r(x)$ with eigenvalue zero.

The rational A$\De$O is related to an A$\De$O associated with the 
dual nonrelativistic nonperiodic  Toda 4-particle system in a similar way as in the $A_2$ case: We have 
\be\label{AATn3}
A_r(x)=i\alpha^{-3}\hat{A}_{-1,\mathrm{nr}}(x)\exp\Big(\frac14 i\alpha \big(-\partial_{x_1}-\partial_{x_2}-\partial_{x_3}-\partial_{x_4}\big)\Big),
\ee
with the A$\De$O  $\hat{A}_{-1,\mathrm{nr}}(x)$  given by Eq.~(4.112) in~\cite{HaRu12} (letting $N=4$ and $\hbar\mu=\alpha$). From Theorem~4.11 in~\cite{HaRu12} it now follows that $A_r(x)$ satisfies the kernel identities
\be
A_r(x)\cT_{\sigma}^{\mathrm{nr}}(x,-y)=A_r(y)\cT_{\sigma}^{\mathrm{nr}}
(x,-y),\ \ \ \sigma=1,-1,
\ee
for unconstrained variables, with dual nonrelativistic kernel functions given by
\be
\cT_{\sigma}^{\mathrm{nr}}(x,y):=\prod_{j,k=1}^4 \Gamma((x_j-y_k)/\alpha)^{\sigma},\ \ \ \ \sigma=\pm 1.
\ee

\newpage

\renewcommand{\thesection}{A}
\setcounter{equation}{0}
 
\addcontentsline{toc}{section}{Appendix A. Generalized gamma functions }

\section*{Appendix A. Generalized gamma functions }
								
We begin this appendix by collecting properties of the elliptic gamma function, cf.~Subsection~III~B in~\cite{Ru97}. It
can be defined by the product representation
\be\label{Gell}
	G(r,a_+,a_-;z) := \prod_{m,n=0}^\infty \frac{1-\exp \big(-(2m+1)ra_+-(2n+1)ra_--2irz\big)}{1-\exp \big(-(2m+1)ra_+-(2n+1)ra_-+2irz\big)}.
\ee
Throughout this paper we choose positive parameters
\be
r,a_+,a_->0.
\ee
This positivity restriction is crucial for quantum-mechanical aspects. (Obviously we only need $\re (ra_{\de})>0$, $\de=+,-$, for \eqref{Gell} to yield a well-defined meromorphic function of~$z$.)
We usually suppress the dependence on the parameters when no ambiguity can occur. Clearly, $G(z)$ is not only meromorphic for $z\in\C$, but also holomorphic and nonzero for $z$ in the strip 
\be\label{strip}
S:= \{z\in\C \mid |\im (z)|<a\}, 
\ee
where we have introduced a new parameter
\be\label{defa}
a:=(a_++a_-)/2.
\ee
 For $z\in S$   we have an alternative representation
  \be\label{Gge}
G(z)=\exp(ig(z)),\ \ \ \ z\in S,
\ee
where   
\be\label{g}
g(z):=\sum_{n=1}^{\infty}\frac{\sin(2nrz)}{2n\sinh(nra_{+})\sinh(nra_{-})},\ \ \ \ z\in S.
\ee
From this  (and also from~\eqref{Gell}), the following properties are obvious:
\be\label{refle}	
G(-z) = 1/G(z),\ \ \ ({\rm reflection\ equation}),
\ee
\be
G(z+\pi/r)=G(z),\ \ \ \ \ ({\rm periodicity}),
\ee
\be\label{modinve}
G(a_-,a_+;z) = G(a_+,a_-;z),\ \ \  ({\rm modular\ invariance}),
\ee
 \be\label{Gcone}
\overline{G(z)}=G(-\overline{z}).
\ee

A slightly different form of the elliptic gamma function is used in most of the recent literature (in particular in Razamat's work~\cite{Ra18}), namely, 
\be\label{Gamell}
\Gamma_e(p,q;x):= \prod_{k,l= 0}^{\infty}\frac{1-x^{-1}p^{k+1}q^{l+1}}{1-xp^{k}q^{l}},\ \ \ \ \ \  |p|<1,\ |q|<1.
\ee
Its relation to~\eqref{Gell} is given by
\be
\Gamma_e(\exp(-2ra_+),\exp(-2ra_-);\exp(-ra_+-ra_-+2irz)) =G(r,a_+,a_-;z).
\ee

The elliptic gamma function $G(z)$ can be viewed as a minimal solution to an analytic difference equation (A$\De$E) that involves a right-hand side function defined by
\be\label{defR}
R(r,\alpha;z) :=\prod_{k=1}^{\infty}[1-\exp(2irz-(2k-1)\alpha r)][1-\exp(-2irz-(2k-1)\alpha r)].
\ee
(This is in essence a rescaled Jacobi theta function.) 
Indeed, letting
\be\label{Rdel}
R_\delta(z):= R(r,a_\delta;z),\ \ \ \ \de=+,-,
\ee
it satisfies the two A$\De$Es
\be\label{Geades}
\frac{G(z+ia_\delta/2)}{G(z-ia_\delta/2)} = R_{-\delta}(z),\quad \delta=+,-,
\ee
with the modular symmetry feature~\eqref{modinve} entailing that only one of the two needs to be verified.

It easily follows from~\eqref{defR} and~\eqref{Rdel} that  the functions $R_+$ and~$R_-$ are entire, even and $\pi/r$-periodic, and $R_{\de}(z)$ satisfies the A$\De$E
\be\label{Rade}
\frac{f(z+ia_\delta/2)}{f(z-ia_\delta/2)} = -\exp(-2irz),\ \ \ \de=+,-.
\ee
They have alternative representations
\be\label{Rrep}
R_{\de}(z)=\exp \left( -\sum_{n=1}^{\infty}\frac{\cos 2nrz}{n\sinh
nra_{\de}}\right),\ \ \ |\im z|<a_{\de}/2,\ \ \ \de=+,-,
\ee
as is readily checked from~\eqref{Gge}, \eqref{g} and~\eqref{Geades}. 

In Section~4 we have occasion to use the functions
\be\label{sR}
s_{\de}(z):= s(r,a_{\de};z)\equiv ie^{-irz}R_{\de}(z-ia_{\de}/2)/p_{\de},\ \  \ \de=+,-,
\ee
where 
\be\label{pde}
p_{\de}:= p(r,a_{\de})\equiv  2r\prod_{k=1}^{\infty}(1-e^{-2kra_{\de}})^2.
\ee
It is easily verified that the function~$s_{\de}(z)$ is entire, odd and $\pi/r$-antiperiodic, and that  it also satisfies the A$\De$E~\eqref{Rade}.
Its relation to the Weierstrass $\sigma$-function is given by
\be\label{sf}
s_{\de}(z)= \exp\big(-\eta(\pi/2r,ia_{\de}/2)\, z^2r/\pi\big)\sigma(z;\pi/2r,ia_{\de}/2),
\ee
where we use the notation of Whittaker/Watson~\cite{WW73}. Using the $G$-A$\De$Es~\eqref{Geades} we readily obtain the A$\De$E
\be\label{Geratio}
\frac{G(z+ia)}{G(z-ia)}=p_+p_-s_+(z)s_-(z),
\ee
which we invoke in Section~4.

 The hyperbolic gamma function can be defined as the unique minimal solution of one of the two A$\De$Es 
\be\label{Ghades}
	\frac{G(z+ia_\delta/2)}{G(z-ia_\delta/2)} = 2\cosh(\pi z/a_{-\de}),\ \ \ \  \delta=+,-,
\ee
that has modulus 1 for real $z$ and satisfies $G(0)=1$ (cf.~Subsection~III~A in~\cite{Ru97}); it is not obvious, but true that this entails that the other one is then satisfied as well.
It is meromorphic in~$z$, 
and for $z$ in the strip $S$~\eqref{strip} 
no poles and zeros occur. Hence we have
\be\label{Ggh}
G(z)=\exp(ig(z)),\ \ \ \ z\in S,
\ee
with $g(z)$ holomorphic in $S$. 
 Explicitly, $g(z)$ has the integral representation
\be\label{ghyp}
g(a_{+},a_{-};z) =\int_0^\infty\frac{dy}{y}\left(\frac{\sin 2yz}{2\sinh(a_{+}y)\sinh(a_{-}y)} - \frac{z}{a_{+}a_{-} y}\right),\ \ \ \ z\in S.
\ee
From this, the following properties of the hyperbolic gamma function are immediate:
\be\label{reflh}	
G(-z) = 1/G(z),\ \ \ ({\rm reflection\ equation}),
\ee
\be\label{modinvh}
G(a_-,a_+;z) = G(a_+,a_-;z),\ \ \  ({\rm modular\ invariance}),
\ee
 \be\label{Gconh}
\overline{G(z)}=G(-\overline{z}).
\ee

The trigonometric gamma function is defined by 
\be\label{Gt}
G_t(r,\alpha;z)\equiv \prod_{n=0}^{\infty} (1-q^{2n+1}\exp(2irz))^{-1},\ \ \ \ q\equiv \exp(-\alpha r),\ \ \ \alpha,r>0,
\ee
It can be viewed as the iteration solution to the A$\De$E
\be\label{Gtade}
\frac{G_t(z+i\alpha/2)}{G_t(z-i\alpha/2)}=1-\exp(2irz),
\ee
cf.~Subsection~III~C in~\cite{Ru97}.
 
 \vspace{1cm}

\noindent
{\Large\bf Acknowledgment}

\vspace{1cm}

\noindent
We are indebted to S.~Razamat, not only for drawing our attention to his paper~\cite{Ra18}, but also for his extensive help in understanding the identities at which he arrived, and  convincing us of their validity.  
\vspace{1cm}

 
\bibliographystyle{amsalpha}

\end{document}